\documentclass[12pt,centertags,reqno]{amsart}

\usepackage{latexsym}
\usepackage[english]{babel}
\usepackage[T1]{fontenc}
\usepackage[numbers]{natbib}
\usepackage{amssymb}
\usepackage{fancyhdr}
\usepackage{url}
\usepackage{hyperref}
\usepackage{verbatim}
\usepackage{leftidx}
\usepackage{epsfig}
 \usepackage{graphicx}
 \usepackage{epstopdf}

\usepackage{mathrsfs, mathtools}
\usepackage{stmaryrd}

\usepackage{marvosym}
\mathtoolsset{showonlyrefs}

\newtheorem{definition}{Definition}[section]

\newtheorem{ass}{Assumption}[section]
\newtheorem{theorem}{Theorem}[section]
\newtheorem{proposition}{Proposition}[section]

\newtheorem{lemma}{Lemma}[section]
\newtheorem{remark}{Remark}[section]

\newcommand{\ud}{\mathrm d}
\newcommand{\ds}{\displaystyle}
\newcommand{\esp}[2][\mathbb E] {#1\left[#2\right]}
\def \I {{\mathbf 1}}

\def \C {\mathcal C}

\def \F {\mathcal F}

\def \N {\mathcal N}

\def \P {\mathbf P}

\def \I {{\mathbf 1}}

\def \R {\mathbb R}
\def \bF {\mathbb F}

\def \bE {\mathbb E}
\def \bN {\mathbb N}

\def \bpi {\boldsymbol{\pi}}

\begin{document}



\title[Portfolio optimization for a large investor]{{Portfolio optimization for a large investor controlling market sentiment  
under partial information}}

\author[S. Altay]{S\"{u}han Altay}
\address{S\"{u}han Altay, \textrm{Department of Financial and Actuarial Mathematics, Vienna University of Technology, Wiedner Hauptstrasse 8-10, 1040 Vienna, Austria}}
\email{altay@fam.tuwien.ac.at}
\author[K.~Colaneri]{Katia Colaneri}
\address{Katia Colaneri, \textrm{Department of Economics,
 University of Perugia, Via Alessandro Pascoli, 20, I-06123 Perugia, Italy.}}\email{katia.colaneri@unipg.it}
\author[Z.~Eksi]{Zehra Eksi}
\address{Zehra Eksi, \textrm{Institute for Statistics and Mathematics, WU-University of Economics and Business, Welthandelsplatz 1, 1020, Vienna, Austria.}}\email{zehra.eksi@wu.ac.at}

\begin{abstract}
 We consider an investor faced with the utility maximization problem in which the risky asset price process has pure-jump dynamics affected by an unobservable continuous-time finite-state Markov chain, the intensity of which can also be controlled by actions of the investor. Using the classical filtering theory, we reduce this problem with partial information to one with full information and solve it for logarithmic and power utility functions. In particular, we apply control theory for piecewise deterministic Markov processes (PDMP) to our problem and derive the optimality equation for the value function and characterize the value function as the unique viscosity solution of the associated dynamic programming equation. Finally, we provide a toy example, where the unobservable state process is driven by a two-state Markov chain, and discuss how investor's ability to control the intensity of the state process affects the optimal portfolio strategies as well as the optimal wealth under both partial and full information cases.

{\textsc{Keywords} : }{\textit{ utility maximization, regime-switching, market sentiment, partial information, piecewise deterministic Markov processes.}}

{\textsc{AMS Classification}:}{\textit{ 60J10,60J75,93E11,93E20}}

\end{abstract}

\maketitle



\section{Introduction}
The influence of large investors, such as hedge funds, mutual funds, and insurance companies, on prices of risky assets, can be studied from very different viewpoints ranging from direct price impact of order execution (selling or buying) to feedback effects from trading to hedge portfolios of derivatives written on the underlying. However, there is also an influence of large investors on the overall market sentiment that arises from their perceived informational superiority. That is, most of the time, the rest of the market takes large investors' portfolio decisions as signals revealing an important insider information not available to small or price-taking investors. Moreover, due to the herding behavior, this effect can be intensified when markets are caught up in certain extreme situations like speculative bubbles or market downturns. Of course, by knowing that they have such an influence on the market, large investors can exploit this fact by changing their portfolio and consumption choices during those times and try to gain an advantage\footnote{For example in the US large institutional investor needs to fill the SEC Form 13F, a form with the Securities and Exchange Commission (SEC) also known as the Information Required of Institutional Investment Managers Form. It is a required form from institutional investment managers with over $100$ million in qualifying assets. It contains information about the investment manager and a list of their recent investment holdings.}. However, it is difficult, even for a large investor, to observe the exact state of the overall market and its effect on the price of the risky asset and hence to act accordingly. Not knowing the exact state of the environment naturally necessitates a partial information setting, in which the large investor only observes the price process of the risky asset.

Therefore in this study, we solve a finite-time utility maximization problem by considering a partially observable regime-switching environment, in which there is a large investor (or group of institutional investors) that has control over the intensity matrix of the continuous-time finite state Markov chain governing the state of the environment. We allow large investor's portfolio choices, as a fraction of the wealth invested in the risky asset, to have an indirect but persistent effect on the price process, through dependence on the controlled intensity of the Markov chain with next-neighborhood-type dynamics. We call this effect {\em market impact}. By taking the generator matrix of the unobservable Markov chain as a function of portfolio holdings of the large investor, and focusing on the price process with pure-jump dynamics affected by the unobservable Markov chain, we solve the problem of utility maximization from terminal wealth for logarithmic and power utility preferences. A similar control problem for optimal investment and consumption for a large investor is studied first by \cite{busch2013optimal} in the full information case with asset prices following jump-diffusion dynamics and a market with two possible states.

To characterize the optimal strategy in a partial information setting, we first solve the corresponding filtering problem and we derive Kushner--Stratonovich type filtering equations using innovations approach. Once the filtering problem is solved, we reduce the optimal control problem under partial information to a full information one, as e.g., in \cite{bauerle2007portfolio,ceci1998partially} where unobservable variables are replaced by their filtered estimates. Since the state of the resulting optimal control problem is piecewise deterministic, we resort to the theory of optimal control for piecewise deterministic Markov processes (PDMP) given in \cite{davis1984piecewise} (see also \cite{davis1993markov} for more details). To be precise, the idea is that the corresponding piecewise deterministic control problem can be recast as a Markov decision process, in which the value function is characterized by a fixed point argument. Here we should note that, although identifying the optimal control problem PDMP with a Markov decision process is well studied (see \cite{davis1993markov, dempster1992necessary, almudevar2001dynamic, forwick2004piecewise,bauerle2010optimal,costa2013continuous} and references therein), to the best of our knowledge, a concrete application of optimal control of PDMP that covers the control of the intensity of an unobservable Markov chain is novel. To this extent, we use modifications of certain results from \cite{colaneri2016shall}, in which the main motivation is to study optimal liquidation in a partial information setting with asset prices having pure-jump dynamics (see also Section \ref{Sec:PDMP} for a deeper discussion). We characterize the value function as the unique fixed point of the reward operator, and further, we obtain a representation in terms of the unique viscosity solution of the Hamilton--Jacobi--Bellman (HJB) equation.

In our setting, the state of the environment with regime-switches can be interpreted in various ways. One natural interpretation, for example in a two-state case, is that states can be characterized as ``bear'' or ``bull'' market sentiments so that the large investor try to change the direction (uptrend or downtrend) of the market by her portfolio choices. Similarly, one may also explain those states as different levels of market ``liquidity'' in a market microstructure framework, or different stages of a business cycle in a more general macrostructure framework. In the former one, a large investor can be seen as a liquidity provider or a market maker, whereas in the latter, she can be considered as a central planner such as a central bank or government\footnote{Certain central banks (Japan and Swiss) around the word have recently invested heavily in stock markets. Although their objective is different than utility maximization from terminal wealth, the same setting (indirectly influencing the economy to give a boost) can be analyzed in the same way.}.

We should also remark that manipulation-type strategies pursued by large investors, in which there is an uncertainty coming from a market reaction against those manipulation attempts, can be modeled in this partially observed control framework. In a similar vein to the credit risk modeling, our modeling framework can be considered as a reduced-form modeling of market manipulation since the impact of large investor on prices is indirect via her influence on market sentiments as opposed to models with direct impact on structural variables such as drift or volatility of the asset price process (see, e.g., \cite{jarrow1992market, jarrow1994derivative} for market manipulation models with large investors having a direct impact on asset price dynamics in discrete and continuous time). In particular, our setting allows for a large investor to change the probability of being in a ``bull'' or ``bear'' market by her actions. For example, by short-selling, a large investor may prevent the market going to a ``bull'' state and hence gain advantage of a ``bear'' market sentiment. Similarly, one can use the proposed model to analyze herding and momentum like behavioural effects on stock prices arising from large investors' portfolio choices, since for example, in our proposed setting with pure-jump type asset price dynamics, one can mimic a market situation, in which a large investor try to influence the market sentiment by changing her portfolio and hence switch it from ``bear'' to ``bull'', where upward jumps are observed more likely, or vice-versa.

Considering the high-frequency nature of the markets that large investors are involved with (see \cite{konikov2002option,elliott2006option} for asset prices with Markov modulated pure-jump dynamics), we should also remark that our choice of working with a pure jump process modulated by a Markov chain is not restrictive.

There is ample amount of literature related to the optimal decision of a large investor, analyzed in various settings. The most related work to ours is \cite{busch2013optimal} that studies the optimal consumption-portfolio choice problem, in which the asset price dynamics are given by a jump-diffusion affected by the regime-switching environment controlled by a large investor in a full information setting. They show that optimal strategies have significant deviations from the strategies obtained in the classical Merton problem. More importantly, they show that there can be situations (market manipulations) such that the large investor can consume even though she has no gain in utility from consumption. Generally, in the literature, the effect of large investors on asset prices are direct in the sense that decision variables (such as portfolio holdings, the speed of trading, etc.) have shown up in the drift or the volatility of the risky asset price process). For instance, the models of \cite{cvitanic1996hedging}, \cite{cuoco1998optimal}, \cite{kraft2011large} and \cite{eksi2017portfolio} examine optimal consumption and investment problem of a large investor with portfolio choices affecting the instantaneous expected returns in various settings. In the context of optimal order execution problems where the stock price process is driven by a diffusion, investors impact is modeled by volume or speed of trading affecting directly the drift (see, e.g., Almgren--Chriss model \cite{almgren2001} and its variants).

There is also a large strand of literature concerning the portfolio optimization problems with Markov modulated price dynamics under partial information. \cite{lakner} and \cite{lakner1998optimal} coonsder the case in which the drift uncertainty is modeled by a linear Gaussian process. \cite{karatzas1998bayesian} has studied the similar problem with a constant but unknown drift. \cite{sass2004} and \cite{haussmann2004} have treated the portfolio optimization problem in a multi-asset setting under partial information, and found the optimal portfolio strategy with martingale approach. On the other hand, \cite{rieder2005portfolio} have addressed the portfolio optimization problem with unobservable Markov chain modulated drift process by using a dynamic programming approach. {\cite{bjork2010optimal} considers a general setting and provides explicit representations of the optimal wealth and investment processes for the utility maximization problem under partial information by using the martingale approach. \cite{frey2012portfolio} solves the portfolio optimization problem under partial information by including expert opinions. Regarding portfolio optimization problems under partial information, one can finally refer to \cite{pham2011portfolio} giving a very broad overview of previous studies on the subject. For the full  information case, there are also studies analyzing portfolio selection problems in a Markov regime-switching framework (see for example \cite{zhou2003markowitz}, \cite{bauerle2004port}, and \cite{sotomayor2009explicit}).

To summarize our contributions, firstly we solve the utility maximization problem for logarithmic and power utility preferences with indirect impact arising from controlling the intensity of the Markov chain both under full and partial information settings. For comparison purposes, we also give solutions to those problems without impact, that is, when there is no control of the intensity. Even for the simple logarithmic utility case, the presence of indirect impact makes point-wise maximization impossible and hence we need to rely on dynamic programming techniques. Secondly, we transform the partial information problem to a full information problem by using stochastic filtering and apply control theory for piecewise deterministic Markov processes (PDMP) to our problem to derive the optimality equation for the value function. We characterize the value function as the unique viscosity solution of the associated dynamic programming equation. Thirdly, by focusing on a two-state Markov chain example, we show that there is always a gain for a large investor from controlling the intensity of the Markov chain both in full and partial settings albeit it is smaller in the latter one. In particular, the large investor can take advantage of the ``bear'' state of the market by short-selling. Also optimal strategies are more aggressive in the presence of market impact such that the large investor buys more in the ``bull'' state and short sells more in the ``bear'' state compared to the corresponding no-impact case. Also it is evident from numerical examples that, as time approaches to the maturity, optimal portfolio strategies with and without impact from intensity control converges to each other under both full and partial information settings.

This paper is structured as follows. In Section \ref{sec:framework}, we introduce the underlying framework as well as the main assumptions used afterwards. In Section \ref{sec:optimization_full_info}, we study the optimization problem under full information and also give the verification result associated to it. Section \ref{sec:optimization_partial_info} contains the optimization problem under partial information and reduction of the problem to full information via stochastic filtering as well as the PDMP techniques to solve the problem and characterization of the optimal value function via unique viscosity solution of the related HJB. Finally, in Section \ref{sec:toy_model}, we present a two-state Markov chain example and discuss model implications for a large investor. We also provide an Appendix, containing technical proofs.

\section{Underlying Framework}\label{sec:framework}
We consider a finite time interval $[0,T]$ and continuous trading in the market. We are given the probability space  $(\Omega,\mathcal{F},\mathbb{F},\mathbb{P})$, where $\mathbb{F}=\{\mathcal{F}_t, \ t \in [0,T]\}$ satisfies the usual conditions; all processes we consider here are assumed to be $\mathbb{F}$-adapted.

We have an investor with a given initial wealth $w\in \mathbb{R}_{> 0}$, and whose objective is to form a self-financing portfolio over the finite period $[0,T]$ in order to maximize the expected utility from terminal wealth by investing in a risky asset and in a risk-free bond. Let $h=\{h_t, \ t \in [0,T]\}$ be the $\bF$-predictable process denoting the fraction of wealth invested in a risky asset. 
Then, $1-h_t$ gives the fraction of the wealth invested in the bond at time $t\in[0,T]$.
We allow for the short-selling of the risky asset and the risk-free bond.  That is, $h_t\in\mathbb{R}$ for every $t \in [0,T]$. We work under

\begin{ass}\label{ass0}
 $h_t \in [-L, L]$, for some $L >0$, for every $t \in [0,T]$.
\end{ass}

This assumption is needed, for example, to deal with set of controls taking values in a compact space. We will see that for the examples considered in Section \ref{sec:toy_model}, this assumption is not restrictive since we obtain an optimal control taking values in $(-L, L)$. We denote by $Y^{(h)}$ a continuous-time finite-state Markov chain representing the state of the market. $Y$ takes values in the canonical state space $\mathcal{E}=\{e_1,e_2,...,e_K\}$ where $e_k$ is the $k$th basis column vector of $\mathbb{R}^K$. The initial distribution of the Markov chain is given by  $\pi_0=(\pi_0^1,\cdots,\pi_0^K)$. The notation $Y^{(h)}$ stands for the fact that we assume that the action of the investor has an impact on the state of the market. Formally, we have that the infinitesimal generator of $Y^{(h)}$ is of the form $Q(h_t)=(q^{i,j}(h_t))_{i,j\in\{1,\dots,K\}}$\footnote{Note that the generator is well defined since $h$ is assumed to be predictable.}. To keep the notation simple, in the following we restrict to the case with next-neighbour dynamics that is $q^{i,j}(\cdot)=0$ if $|i-j|>1$. However note that results can be easily extended to the general one. This implies the following structure for the generator

{\small\begin{equation*} \label{eq:generator_MC}
Q(h_t)=\left(
\begin{array}{cccccc}
-q^{1,2}(h_t)&q^{1,2}(h_t)&0&\dots&0&0\\
q^{2,1}(h_t)&-q^{2,1}(h_t)-q^{2,3}(h_t)&q^{2,3}(h_t)&\dots&0&0\\
\vdots&\vdots&\vdots&\ddots&\vdots&\vdots\\
0&0&0&\dots&q^{K,K-1}(h_t)&-q^{K,K-1}(h_t)
\end{array}
\right)
\end{equation*}}

\noindent where  $q^{i,j}:[-L,L]\to\R_{\ge 0}$ is a nonnegative continuous function for $i \neq j$ and $i,j \in \{1, \dots, K\}$.  

We consider a risk-free bond and a risky asset as available instruments in the market, with price processes $B=\{B_t, \ t \geq 0\}$ and $S=\{S_t, \ t \ge 0\}$, respectively. The bond price is assumed to follow
\begin{equation*}\label{bond}
\ud B_t=\rho B_t\ud t, \quad B_0\in \R_{>0},
\end{equation*}
where $\rho>0$ is the instantaneous risk-free rate.

The risky asset price process has pure-jump dynamics that is affected from the state of the market. Formally, it evolves according to the following equation:
\begin{equation*}\label{eq:bid_price}
\ud S_t=S_{t^-}\int_{\mathcal Z} G(t,Y^{(h)}_{t^-}, \zeta) \N(\ud t, \ud \zeta),\quad S_0\in \R_{>0},
\end{equation*}
where  $\N(\ud t, \ud \zeta)$ is a Poisson random measure on $\R_{\geq 0}\times \mathcal Z$, with $\mathcal Z\subseteq \R$, having finite intensity $\varsigma(\ud \zeta)\ud t$ independent of the Markov chain $Y^{(h)}$, and $G:[0,T]\times \mathcal{E}\times \mathcal Z \to \R$ is a measurable function, continuous in time and satisfying $$\mathbb{E}\left[\int_0^T\int_{\mathcal Z} G^2(t,Y^{(h)}_{t^-}, \zeta)\varsigma(\ud \zeta)\ud t \right]<\infty.$$ In order to ensure the non-negativity of the process $S$ we further assume that $1+ G(t,e_i,\zeta)>0$ for every $(t,\zeta)\in [0,T]\times \mathcal Z$ and $i\in \{1, \dots, K\}$ and moreover we assume that equation~\eqref{eq:bid_price} has a unique solution. A set of sufficient conditions for uniqueness of the solution is given, for example, in \cite[Theorem 1.19]{oksendal2005applied}.

Let  $R:=\{R_t, \ t \in [0,T]\}$ be the return process,
\begin{equation*}\label{eq:return}
\ud R_t=\int_{\mathcal Z} G(t,Y^{(h)}_{t^-}, \zeta) \N(\ud t, \ud \zeta),
\end{equation*}
and introduce the random measure $\mu(\ud t,\ud z)$ associated to its jumps
\[
\mu(\ud t, \ud z):=\sum_{s:\Delta R_s\neq 0}\I_{\{s, \Delta R_s\}}(\ud t, \ud z).
\]

Then the following equality holds
\begin{equation*}\label{eq:return2}
R_t=\int_0^t\int_{\mathcal Z}  \ G(t, Y^{(h)}_{t^-},\zeta) \ \N(\ud t, \ud \zeta)=\int_0^t\int_\R z \ \mu(\ud t, \ud z),
\end{equation*}
for every $t \in [0,T]$. We denote the $(\bF,\P)$-dual predictable projection of the measure $\mu$ by $\eta^{\P}(t, Y_{t^-}^{(h)},\ud z)\ud t$. For every $A\in\mathcal{B}(\R)$ the following holds
\[
\eta^{\P}(t, Y_{t^-}^{(h)}, A)=\varsigma(D^A_t),
\]
where $D^A_t:=\{\zeta \in \mathcal Z: \ G(t,Y^{(h)}_{t^-},\zeta)\in A\setminus\{0\}  \}$, see, e.g. \cite[Chapter 8]{bremaud1981point}. 
The assumptions on $G$ and $\varsigma$ imply that $\eta^\P(t,e_i,z)$ is continuous in time and that
\begin{equation}\label{eq:integrability}\esp{\int_0^T\int_{\R}z^2\eta^{\P}(t, Y_{t^-}^{(h)},\ud z)\ud t}< \infty. \end{equation}

Let $W^{(h)}=\{W^{(h)}_t, \ t \in [0,T]\}$ be the wealth process corresponding to a self-financing strategy $h=\{h_t,  t \in [0,T]\}$. The dynamics of  $W^{(h)}$ is given by
\begin{equation}\label{wealthdyn}
\ud W^{(h)}_t=W^{(h)}_{t^-}\left((1-h_t)\rho \ud t+ h_{t}\int_{\R}z\mu(\ud t, \ud z) \right), \quad W^{(h)}_0\in \mathbb{R}_{>0}.
\end{equation}

In order to ensure that the wealth process is positive we consider investment strategies that satisfy
\begin{ass}\label{ass2}
$\eta^{\P}(t, e_i, \Theta)=0$  for every $t \in [0,T]$ and $i\in\{1,\dots,K\}$, where  $\Theta=\{z \in \R:\ 1+h_tz\le0, \ t \in [0,T] \}$.
\end{ass}

Note that this assumption can be weakened if, for instance, short selling is prohibited.

In the sequel, whenever there is no ambiguity, for the sake of notational ease, we suppress the dependence of the processes $Y^{(h)}$ and  $W^{(h)}$ on the strategy $h$, and simply write $Y$ and $W$. Finally, we can write the solution for \eqref{wealthdyn}
\begin{align*}
W_t=W_0\exp\left\{\int_0^t\left((1-h_s)\rho+\int_{\R}\log(1+h_s z)\eta^{\P}(s, Y_{s^-}\ud z)\right) \ud s \right. \\
+\left.\int_0^t\int_{\R} \log(1+h_{s} z)\nu(\ud s, \ud z)  \right\},   \nonumber
\end{align*}
for every $t \in [0,T]$, where
\begin{equation}\label{eq:nu}
\nu(\ud t,\ud z):=\mu(\ud t,\ud z)-\eta^{\P}(t, Y_{t^-},\ud z)\ud t
\end{equation}
indicates the compensated jump measure associated with $\mu$.


In the rest of the paper, we always work under the \textit{standing assumptions} made in Section \ref{sec:framework}.  

\section{Optimization Problem under Full Information}\label{sec:optimization_full_info}
In the first step we assume that the investor has the full knowledge of the market. Formally this means that the available information is given by the filtration $\bF$.
This leads to the following definition of admissible strategies.

\begin{definition}
  A portfolio strategy $h$ is $\bF$-admissible if it is $\bF$-predictable and Assumptions \ref{ass0} and \ref{ass2} hold. We denote the set of $\bF$-admissible strategies by $\mathcal H$.
\end{definition}

Suppose we are given a strictly increasing, strictly concave and continuously differentiable utility function $U:\mathbb{R}_{>0}\rightarrow \mathbb{R}$ satisfying Inada conditions, i.e. $\lim_{w\to 0}\frac{\partial U}{\partial w}(w)=\infty$ and $\lim_{w\to \infty}\frac{\partial U}{\partial w}(w)=0$.
The goal of the investor is to solve the following optimization problem
\begin{equation}\label{opt_full}
\max\,\mathbb{E}^{t,w,i}\left[ U(W_T)\right],
\end{equation}
over all admissible strategies, subject to the initial value of the wealth $W_t=w$ and initial state $Y_t=e_i$ for some $i \in \{1, \dots, K\}$.


The value function for the current optimization problem is
\begin{align*}\label{vf}V(t,w,e_i)=\sup_{h \in \mathcal{H}}\mathbb{E}^{t,w,i}[U(W_T)].\end{align*}
If $V$ is continuous and differentiable with respect to the first two arguments, i.e. $V \in \C^{1,1}_b([0,T] \times \R_{>0}\times\mathcal{E})$, then it can be characterized as the unique classical solution of the HJB equation given by
\begin{align*}
0=\sup_{h \in [-L,L]}\left\{\mathcal L^{h} V(t,w,e_i)\right\},
\end{align*}
with $\mathcal L^h$ being the $(\bF, \P)$-Markov generator of the pair $(W,Y)$. Explicitly, we have
 \begin{align}
  0=& \sup_{h \in [-L,L]} \Bigg\{ \frac{\partial V}{\partial t}(t,w,e_i)+ \frac{\partial V}{\partial w}(t,w,e_i) w(1-h)\rho+ \sum_{j=1}^K (V(t,w,e_j)-V(t,w,e_i)) q^{i,j}(h) \\
 &+\int_\R \left[V\left(t,w(1+h z), e_i\right)-V(t,w,e_i)\right] \eta^{\P}(t,e_i,\ud z)\Bigg\},\label{eq:HJB_1}
\end{align}
with the final condition $V(T, w, e_i)=U(w)$, for every $w\in \R_{>0}$ and $i\in\{1,\dots,K\}$.
In the next theorem, by combining classical results we provide a verification result for the optimization problem \eqref{opt_full}.
\begin{theorem}\label{thm:verification}
Let $\Upsilon$ be a solution to equation \eqref{eq:HJB_1}, and assume that for every control $h\in \mathcal H$  the following conditions hold
\begin{align}
&\mathbb{E}\Bigg[\int_0^T \!\!\int_\R |\Upsilon\left(s,W^{(h)}_{s^-}(1+h_{s} z), Y_{s^-}\right)-\Upsilon(s,W^{(h)}_{s^-}, Y_{s^-})| \eta^{\P}(\ud s,Y_{s^-},\ud z)\ud s \Bigg]<\infty,\label{eq:cond1}\\
&\mathbb{E}\Bigg[\int_0^T \sum_{j=1}^K \sum_{k \neq j}|\Upsilon(s, W^{(h)}_{s^-},  e_j)-\Upsilon(s, W^{(h)}_{s^-}, Y_{s^-})| q^{k,j}(h_s)\I_{\{Y_{s^-}=e_k\}}\ud s \Bigg]<\infty.\label{eq:cond2}
\end{align}
Then,
  \begin{itemize}
  \item[i.] $\Upsilon(t,w,e_i)\geq \mathbb{E}^{t,w,i}\big[U(W^{(h)}_T)\big]$, for all  $(t,w,e_i)\in [0,T]\times \mathbb R_{>0}\times \mathcal E$;
  \item[ii.] if there exist a strategy $h^*\in \mathcal H$ such that
  \[
  h^*_s\in\underset{h\in [-L,L]}{\arg\max}\left[ \mathcal{L}^h \Upsilon(s,W^{(h^*)}_s,Y^{(h^*)}_s)   \right] \quad \P-a.s.
  \]
  for every $s \in [0,T]$, then $\Upsilon(t,w,e_i)=V(t,w,e_i)$. Moreover $h^*$ is an optimal portfolio strategy.
  \end{itemize}
  \end{theorem}
The proof is in Appendix \ref{app:proofs}.

In the sequel, we deal with the utility maximization problem for an investor who has logarithmic and power utility preferences, respectively.

\subsection{Logarithmic utility}\label{sec:log_utility_full_info}

We consider the portfolio optimization problem for a large investor with logarithmic utility preference. That is, we have $U(w)=\log(w)$.
For comparison purposes, we first study the degenerate case,  where the generator of the Markov chain does not depend on the actions of the investor. This is the case where the investor has no market impact. Then we move to our primary interest, the case with market impact.
Normally the logarithmic utility is the simplest case and can be solved by pointwise maximization. However, with the inclusion of the market impact this is not possible anymore since the current actions of the investor have an influence on the future states of the market and therefore may change the jump intensity of the asset price process.

\subsubsection{Logarithmic utility - No market impact}
To begin with we provide a characterization of the optimal strategy and a stochastic representation for the value function in the setting where the intensity of the Markov chain does not depend on the portfolio strategy, that is, $q^{i,j}(h)\equiv q^{i,j}$, for $i,j\in\{1,\dots,K\}$ and every control $h$. In this case the optimal control problem can be solved directly. First note that by applying the It\^{o}'s formula we get
\[
V(t,w,e_i)=\log(w)+\sup_{h \in \mathcal H} \widetilde \beta(t,e_i;h),
\]
where
\begin{align*}
\widetilde \beta(t, e_i;h)=\bE^{t,i}\left[\int_t^T\left((1-h_s)\rho+\int_{R}\log(1+h_s z) \eta^{\P}(s,Y_{s^-},\ud z)\right)\ud s \right. \\
\left.+\int_t^T\int_{\R} \log(1+h_{s} z)\nu(\ud s, \ud z)\right].
\end{align*}

\begin{proposition}\label{logu2} Suppose $U(w)=\log(w)$ for $w>0$.
\begin{itemize}
\item[i)] Let $h^*(t, e_i)$ 
satisfy either
\begin{align}\label{eq:logoptimalna}
\int_{\R}\frac{z}{1+h^{\ast}(t, e_i) z} \eta^{\P}(t,e_i,\ud z)=\rho
\end{align}
or $h^{\ast}(t,e_i)\in\{-L,L\}$, for every $i\in\{1,\dots,K\}$. Then the optimal strategy $h^*_t=h^*(t, e_i)$ for every $t \in [0,T]$ and $i \in \{1, \dots, K\}$.
\item[ii)] The value function is of the form
\end{itemize}
\begin{equation*}\label{eq:log_utility_full}
V(t,w ,e_i)=\log(w)+\bE^{t,i}\left[\int_t^T\left((1-h^{{\ast}}_s)\rho+\int_{R}\log(1+h^{{\ast}}_{s} z) \eta^{\P}(s,Y_{s^-},\ud z)\right)\ud s \right].
\end{equation*}
\end{proposition}
\begin{proof}
We start by writing
\begin{align}\log(W_T)&=\log(w)+\int_t^T\left((1-h_s)\rho+\int_{R}\log(1+h_s z) \eta^{\P}(s,Y_{s},\ud z)\right)\ud s \\
+&\int_t^T\int_{\R} \log(1+h_{s} z)\nu(\ud s, \ud z).\label{eq:utility}\end{align}
It follows from condition \eqref{eq:integrability} that the process  $$\int_0^t\int_{\R} \log(1+h_{s} z)\nu(\ud s, \ud z), \ t \in [0,T],$$ is an $(\bF, \P)$-true martingale. Then, by taking the expectation on both sides of \eqref{eq:utility} we have
\[\mathbb{E}^{t,w,i}[\log(W_T)]=\log(w)+\mathbb{E}^{t,i}\bigg[\int_t^T \left((1-h_s)\rho+\int_{R}\log(1+h_s z) \eta^{\P}(s,Y_{s},\ud z)\right)\ud s\bigg].\]
Now we can maximize pointwisely. Since $Y$ is independent of the control, at time $t$ we get the first order condition
\begin{equation}0=-\rho+\int_{\R}\frac{z}{1+h_tz}\eta^{\P}(t,e_i,\ud z).\label{eq:foc}\end{equation}
Provided that equation \eqref{eq:foc} has a solution $h^{\ast}(t, e_i)$, the second order condition
$$
-\int_{\R}\frac{z^2}{(1+h_tz)^2}\eta^{\P}(t,e_i,\ud z)<0
$$
implies that this is the global maximizer. Otherwise the maximum is attained at one of the boundary points $\{-L,L\}$.
\end{proof}
We remark that, for the case study of Section \ref{sec:toy_model}, equation \eqref{eq:foc} admits always an interior solution $h^*(t, e_i)\in (-L,L)$.

An extensive study of the utility maximization  with logarithmic preferences in the classical case (i.e. without market impact) is given by \cite{goll_kallsen}, where the optimal strategy is characterized in terms of the local characteristics (drift, volatility and jump intensity) of the semimartingale driving the asset price process (see, e.g., \cite[Theorem 3.1]{goll_kallsen}).

\subsubsection{Logarithmic utility - Market impact}
In the case with market impact the above procedure does not apply. This is due to the fact that at any point in time the decision of the investor may change the future state of the Markov chain. Therefore here we address the problem via dynamic programming. Precisely, we study the solution to equation \eqref{eq:HJB_1} with the terminal condition $V(T,w, e_i)=\log (w)$. We consider the ansatz
$V(t,w,e_i)=\log (w)+\beta(t,e_i)$. Then we have the following system of equations for $(t,e_i)$:
\begin{align}
 -\frac{\partial \beta}{\partial t}(t,e_i)=& \sup_{h \in [-L,L]} \!\Big\{ \! (1-h)\rho + (\beta(t, e_{i+1})-\beta(t,e_i)) q^{i, i+1}(h)  \\
 &+(\beta(t, e_{i-1})-\beta(t,e_i)) q^{i, i-1}(h) + \int_\R \log(1+h z) \eta^{\P}(t,e_i,\ud z)\Big\}.\label{eq:hjblog}
\end{align}
for every $t \in [0,T]$ and $i\in\{2, \dots,K-1\}$, and

\begin{align}
\frac{\ud \beta}{\ud t}(t,e_1)=& - \sup_{h \in [-L,L]} \Big\{    (1-h) \rho  + \left(\beta(t,e_2)-\beta(t,e_1)\right) q^{1,2}(h) \\
&+\int_\R \log(1+h z) \eta^{\P}(t,e_1,\ud z)\Big\},\label{eq:hjblog1}
\end{align}

\begin{align}
\frac{\ud \beta}{\ud t}(t,e_K)=& - \sup_{h \in [-L,L]} \Big\{    (1-h) \rho  +\left(\beta(t, e_{K-1})-\beta(t,e_K)\right) q^{K,K-1}(h) \\
&+\int_\R \log(1+h z) \eta^{\P}(t,e_K,\ud z)\Big\},\label{eq:hjblogK}
\end{align}
respectively, with  boundary conditions $\beta(T, e_i)=0$ for $i\in\{1,\dots, K\}$. Equations \eqref{eq:hjblog},\eqref{eq:hjblog1} and \eqref{eq:hjblogK}  imply that given an optimizer $h^*$, $\beta(t, e_i)$, $i\in\{1,\dots,K\}$, is the unique solution of this system of ordinary differential equations (ODEs). This follows from the continuity of the coefficients \cite[Theorem 3.9]{teschl2012ordinary}.  In principle, one can solve the system numerically using, for instance, backward Euler method. In particular, as pointed out in \cite{busch2013optimal}, at each time step $t_n$ of the numerical procedure one should find the maximizer $h^*(t_n)$, and then solve the resulting ODE.

To verify that the solution of equations \eqref{eq:hjblog}, \eqref{eq:hjblog1}, and \eqref{eq:hjblogK} are indeed the value function for the current optimization problem we make the observation that
for every $h\in \mathcal H$ and for every $i\in\{1,\dots,K\}$,
\begin{align*}
&\esp{\int_0^T\int_{\R}\log(1+h_tz)\eta^{\P}(t, Y_{t},z)\ud t}< \infty,\\
&\esp{\int_0^T\sum_{j=1}^K(\beta(t, e_j)-\beta(t, e_i))q^{i,j}(h_t)\ud t}<\infty,
\end{align*}
where the first inequality follows from condition \eqref{eq:integrability}, and the second one is clear from boundedness of $q^{i,j}(h)$. Then verification Theorem \ref{thm:verification} applies.

\subsection{Power utility}\label{sec:power_utility_full_info}
In this part we work under the assumption of power utility, that is, $U(w)=\frac{1}{\theta}w^{\theta}$, $\theta<1$, $\theta \neq 0$. We address the corresponding optimization problem by dynamic programming technique. In what follows we investigate the solution to the equation \eqref{eq:HJB_1} with the terminal condition $V(T,w, e_i)=\frac{w^\theta}{\theta}$. To this, we suggest the following ansatz for the value function
\begin{align}
V(t, w, e_i)&=\frac{w^\theta}{\theta} e^{\theta \gamma(t,e_i)}, \quad i\in\{1,\dots,K\}. \label{eq:ansatz1}
\end{align}
Inserting \eqref{eq:ansatz1} into \eqref{eq:HJB_1} leads to equations
\begin{align}
&\frac{\ud \gamma}{\ud t}(t,e_i)= - \sup_{h \in [-L,L]} \left\{    (1-h) \rho  +\frac{1}{\theta} \left(e^{\theta (\gamma(t, e_{i-1})-\gamma(t,e_i))}-1\right) q^{i, i-1}(h)\right.\nonumber\\
&\qquad+\frac{1}{\theta}\left. \left(e^{\theta (\gamma(t, e_{i+1})-\gamma(t,e_i))}-1\right) q^{i, i+1}(h)+\frac{1}{\theta}\int_\R \left((1+h z)^\theta-1\right) \eta^{\P}(t,e_i,\ud z)\right\} \label{eq:hjbpower}
\end{align}
for every $t \in [0,T]$ and $i\in\{2, \dots,K-1\}$, and

\begin{align}
\frac{\ud \gamma}{\ud t}(t,e_1)=& - \sup_{h \in [-L,L]} \left\{    (1-h) \rho  +\frac{1}{\theta} \left(e^{\theta (\gamma(t,e_2)-\gamma(t,e_1))}-1\right) q^{1,2}(h) \right.\nonumber\\
&\left.+\frac{1}{\theta}\int_\R \left((1+h z)^\theta-1\right) \eta^{\P}(t,e_1,\ud z)\right\},\label{eq:hjbpower1}
\end{align}

\begin{align}
\frac{\ud \gamma}{\ud t}(t, e_K)=& - \sup_{h \in [-L,L]} \left\{    (1-h) \rho  +\frac{1}{\theta} \left(e^{\theta (\gamma(t, e_{K-1})-\gamma(t, e_K))}-1\right) q^{K, K-1} (h)\right.\nonumber\\
&+\left.\frac{1}{\theta}\int_\R \left((1+h z)^\theta-1\right) \eta^{\P}(t,e_K,\ud z)\right\},\label{eq:hjbpowerK}
\end{align}
respectively, with final conditions $\gamma(T, e_i)=0$ for $i\in\{1,\dots, K\}$. Given an optimizer $h^*$, $\gamma(t,e_i)$, every $i\in\{1,\dots,K\}$, is the unique solution of the system of first order ODEs given by equations \eqref{eq:hjbpower},\eqref{eq:hjbpower1} and \eqref{eq:hjbpowerK}. Note that a simple transformation, i.e., $F(t, e_i)=e^{\theta \gamma(t,e_i)}$, yields to a system of linear ODEs. One can follow the same procedure as in the case of logarithmic utility  and solve the system numerically.

Moreover by the boundedness of $q^{i,j}(h)$ and condition \eqref{eq:integrability}, we have that for every $h\in \mathcal H$,
\begin{align*}
&\esp{\int_0^TW_t^\theta\int_{\R}(1+h_tz)^\theta \eta^{\P}(t, Y_{t},z)\ud t}< \infty,\\
&\esp{\int_0^TW_t^\theta\sum_{j=1}^K |e^{\theta \gamma(t, e_j)}- e^{\theta \gamma(t,e_i)}|q^{i,j}(h_t)\ud t}<\infty,
\end{align*}
for every $i\in\{1,\dots,K\}$, and hence, the verification Theorem \ref{thm:verification} holds.



\section{Optimization Problem under Partial Information}\label{sec:optimization_partial_info}
In the current section we assume that the state process $Y$ is not directly observable by the investor. Instead, she observes the price process $S$ and knows the model parameters. Hence, the available information is represented by the natural filtration generated by the risky asset price process,
\begin{equation*}\label{eq:small_filtration}
\bF^S:=\{\F^S_t, t \in [0,T]\}, \quad \F^S_t:=\sigma\{S_s, \ 0\leq s\leq t\}, \quad \forall t \in [0,T].
\end{equation*}
Throughout the paper we assume that $\bF^S$ satisfies the usual conditions.

At any time $t\in [0,T]$ the decision of the investor depends only on the available information. Accordingly, we define the set of admissible strategies as follows.
\begin{definition}
  A portfolio strategy $h$ is $\bF^S$-admissible if it is $\bF^S$-predictable and assumptions \ref{ass0} and \ref{ass2} hold. We denote the set of $\bF^S$-admissible strategies by $\widetilde{\mathcal H}$.
\end{definition}

Considering $\mathbb{F}^S$-predictable investment strategies results in an optimal control problem under partial information. In a Markovian setting as the one outlined here we can reduce the control problem under partial information to an equivalent control problem under full information where the unobservable state variable, namely the Markov chain $Y$, is replaced by the filtered estimates, see, for example, \cite{bauerle2007portfolio,ceci1998partially}. This requires to solve a filtering problem where the unobservable signal is given by the Markov chain $Y$ and the observation process is the pure jump process $S$.
The literature on filtering problem with pure jump process observation is relatively large. A brief list of results includes for instance \cite{bremaud1981point,ceci2006, bib:elliott-malcolm-08, frey2012pricing, ceci2014zakai}.
In the upcoming part, we deal with the filtering problem corresponding to our setting by using the so called innovations approach. A similar problem is solved in \cite{colaneri2016shall} with a different methodology. In that paper,  the dependence of the jump intensity of the stock price on the control lead to circularity of information which made it not possible to use the innovation approach. Instead they use the reference probability method.

\subsection{Filtering and reduction to full information}\label{sec:filtering}
Define the filter $\pi(f):=\{\pi_t(f), t \in [0,T]\}$  by
\begin{equation*}\label{eq:filter}
\pi_t(f)=\esp{f(Y^{(h)}_t)|\F^S_t}, \quad t \in [0,T],
\end{equation*}
for every function $f:\mathcal{E}\to \R$ and every control $h$. Note that we suppress the dependence of $\pi$ on $h$ for the ease of notation. Denote by $\pi_{t^-}(f)$ the predictable version of the filter. We define
\[\pi^i_t:=\esp{\I_{\{Y_t^{(h)}=e_i\}}|\F^S_t}, \quad t \in [0,T], \]
the conditional state probabilities of the Markov chain $Y^{(h)}$, for every fixed strategy $h$. Since $\mathcal E$ is finite, we may write
\begin{equation*}\label{eq:filterMC}\pi_{t}(f)=\sum_{j=1}^K f(e_j)\pi^j_t, \quad t \in [0,T].\end{equation*}
A fundamental step for applying the innovation approach is to write a representation for $(\bF^S, \P)$-martingales. We introduce the following notation
\begin{equation}\label{eq:dpp}\pi_{t^-}(\eta^\P(\ud z))\ud t:=\sum_{i=1}^K\pi^i_{t^-}\eta^\P(t,e_i,\ud z)\ud t.\end{equation}
It is not difficult to show that for every nonnegative $(\bF^S, \P)$-predictable process indexed by $z$, $\Phi:=\{\Phi(t,z), \ t \in [0,T]\}$ such that $$\esp{\int_0^T \int_\R |\Phi(s,z)| \pi_{t^-}(\eta^\P(\ud z))\ud t}<\infty,$$ the following holds (see, \cite[V T28]{dm2}):
\[
\esp{\int_0^T\int_\R\Phi(t,z) \mu(\ud t, \ud z)}=\esp{\int_0^T\int_\R\Phi(t,z)\sum_{i=1}^K\pi^i_{t^-}\eta^\P(t,e_i,\ud z)\ud t},
\]
which implies that \eqref{eq:dpp} provides the  $(\bF^S, \P)$-dual predictable projection of the measure $\mu$ and that the process
$\int_0^t\int_{\R} \Phi(s,z)\left(\mu(\ud s,\ud z)-\pi_{s^-}(\eta^\P(\ud z))\ud s\right) $, $t \in [0,T]$, is an $(\bF^S, \P)$-martingale.

Let ${\nu}^\pi(\ud t,\ud z)$ denote the $(\bF^S, \P)$-compensated measure, that is
\begin{equation}\label{eq:compensator_partial}
{\nu}^\pi(\ud t,\ud z):=\mu(\ud t, \ud z)-\pi_{t^-}(\eta^\P(\ud z))\ud t.
\end{equation}
This is the building block for the innovations process.


\begin{proposition}\label{prop:KS}
The process $\pi^i$, for all $i\in\{1,\dots,K\}$ solves the equation
\begin{equation}\label{eq:KS_I}
\ud \pi_t^i= \sum_{j=1}^Kq^{j,i}(h_t)\pi_{t}^j \ud t+\int_{\mathbb{R}} \pi_{t^-}^i u^i(t,\pi_{t^-},z){\nu}^\pi(\ud t, \ud z),
\end{equation}
where $\ds u^i(t,\pi_t,z):=\frac{1}{\sum_{j=1}^K \pi_t^j \frac{\ud \eta^{\P}(t,e_j,z)}{\ud \eta^{\P}(t,e_i,z)}}-1$ and $\ds \frac{\ud \eta^{\P}(t,e_j,z)}{\ud \eta^{\P}(t,e_i,z)}$ denotes the Radon-Nikodym derivative of the measure $\eta^{\P}(t,e_j,\ud z)$ with respect to $\eta^{\P}(t,e_i,\ud z)$.
\end{proposition}
The proof is postponed to Appendix \ref{app:proofs}.

Uniqueness of the solution of the filtering equation is necessary to transform the optimal control problem stated in \eqref{eq:optimization_full} into an equivalent one involving only observable processes. Therefore, in the rest of the paper we assume that the Kushner-Stratonovich (KS) equation has a unique solution.

\begin{remark}
A sufficient condition for uniqueness of the solution of (KS) equation is, for example, $$\sup_{t\in[0,T]}\eta^\P(t,e_i, \R)<\infty,$$ for every $i\in \{1,\dots, K\}$, see, e.g., \cite{ceci2014zakai}. This is satisfied in our model since $\varsigma$ is a finite measure.  
\end{remark}
Note that the asset price $S$ as well as the wealth process $W$ have a representation with respect to investor's information, given by
\begin{align*}
S_t=&S_0+\int_0^t \sum_{i=1}^K \int_{\R} z S_{s} \pi^i_s \eta^\P(s, e_i, \ud z)\ud s+ +\int_0^t \int_{\R} z S_{s} \nu^\pi(\ud t, \ud z),\\
W_t=& W
_0 + \int_0^t W_{s} \left((1-h_s) \rho + \sum_{i=1}^K \int_{\R}z \eta^\P(s, e_i, \ud z)\right)\ud t+ W_{s^-} h_{s} \int_{\R}z \nu^\pi(\ud t, \ud z),
\end{align*}
for every $t \in [0,T]$.
In the partial information framework we can write the objective of the investor as
\begin{equation}\label{eq:optimization_full}
\max \mathbb{E}^{t,w,\bpi}\left[ U(W_T)\right],
\end{equation}
over the set of $\bF^S$-admissible controls, where $\mathbb E^{t,w,\bpi}$ denotes the conditional expectation given $W_t=w$ and $\pi_t=\bpi$. The control problem  is characterized by the $(K+1)$-dimensional state process $(W,\pi)$ where $\pi$ is the vector process $(\pi^1,\dots,\pi^K)$ which takes values on the $(K-1)$-dimensional simplex $\Delta_K$. 
We define the reward and the value functions as
\begin{align*}
\nonumber J(t,w,\bpi;h)&=\mathbb{E}^{t,w,\bpi}\left[U(W_T)\right],\\
\label{eq:optimization2} V(t,w,\bpi)&=\sup_{h\in \widetilde{\mathcal{H}}}J(t,w,\bpi;h).
\end{align*}

\subsection{Solution via piecewise deterministic Markov processes approach}\label{Sec:PDMP}

The state process of the optimization problem, consisting of the wealth process and the filter, augmented by the time variable, is a piecewise deterministic Markov process (PDMP), in the sense of \cite{davis1993markov}. A PDMP is a combination of a deterministic flow, characterized as the solution of an ordinary differential equation, and random jumps.

To identify the proper structure of the problem and the appropriate conditions to apply the theory of control for PDMP, we start by introducing some notation.  Let $\mathcal{X}=\R_{>0}\times \Delta_K$ be the state space and $\widetilde{\mathcal X}=[0,T]\times \R_{>0} \times \Delta_K$ the augmented one and denote the state process and the augmented state process by $X:=(W,\pi)$ and $\widetilde X:=(t,W,\pi)$ respectively. Denote by $\{T_n\}_{n \in \mathbb N}$ the sequence of jump times of the state process $\widetilde X$. Then between two consecutive jump times before time $T$, i.e. $ t \in [T_n\wedge T, T_{n+1}\wedge T)$, the state process $\widetilde X$ is described by the ODE $\ud \widetilde X_t=g(\widetilde X_t, h_t) \ud t$, where the vector field $g:\widetilde{\mathcal X}\times [-L,L]\to \R$ is given by
\begin{gather*}
g^{(1)}(\widetilde x, h)=1, \quad g^{(2)}(\widetilde x, h)= w(1-h)\rho, \\
 g^{(i+2)}(\widetilde x, h)=\sum_{j=1}^K\pi^j\left(q^{j,i}(h) +\int_{\mathbb{R}} \pi^i u^i(t,z)\eta^\P(t, e_j, \ud z)\right), \quad i\in \{1, \dots, K\}.
\end{gather*}
The jump rate of the state process is given by $\lambda (\widetilde X)$, where
\begin{align*}
\lambda(\widetilde x)=\lambda(t,w,\bpi)=\sum_{i=1}^K \pi^i \eta^\P(t, e_i, \R),
\end{align*}
and it is independent of $w$. According to \cite{davis1993markov}, the transition kernel that governs the jumps of the state process is described by the operator $Q$
{\small \begin{align*} 
& Q_{\widetilde X} f (\widetilde x, h )  := \int_{\widetilde{\mathcal X}}f(\widetilde y)Q_{\widetilde X}(\ud \widetilde y \mid\widetilde x,h)  \\
 &= \lambda(\widetilde x)\!\sum_{j=1}^K \!\pi^j \!\int_{\mathbb{R}}\! f\left(t,w(1+hz),\pi^1 (1+u^{1}(t,\bpi,z)),\dots, \pi^{K}(1+u^{K}(t,\bpi,z)) \right) \eta^\P(t,e_j,\ud z) \,.
\end{align*}}

Now we define the Markov policy. Denote by $\mathcal A$ the set of measurable mappings $\alpha:[0,T]\to [-L,L]$ and  define an  \emph{admissible strategy} as a sequence of mappings $\{h^n\}_{n \in \mathbb{N}}:\widetilde{\mathcal X}\to \mathcal A$, where the portfolio weight at time $t$ is given by
\begin{equation}\label{eq:admissible-strategy-1}
h_t=\sum_{n \in \bN} \I_{(T_n\wedge T, T_{n+1}\wedge T]} (t)   h^n \left(t- T_n,\widetilde X_{T_n}\right).
\end{equation}
Since at any jump time the evolution of the state process is known up to the next jump time, the idea is that an optimal investment strategy consists of a sequence of choices $h^n$ taken at each jump time $T_n<T$ and to be followed up to $T_{n+1}\wedge T$.
Note that although in the most general form of admissible strategies $h^n$ should depend on the whole past history (see \cite[Theorem T34, Appendix A2]{bremaud1981point}), this larger class of policies does not increase the value of the control problem. This means that we can restrict to consider admissible strategies of the form \eqref{eq:admissible-strategy-1}.

Denote by $\P^{\{h^n\}}_{(t,x)}$ (equiv.~$\P^{\{h^n\}}_{\widetilde x}$) the law of the state process provided that  $X_t=x\in\mathcal{X}$ and that the investor uses the strategy $\{h^n\}_{n \in \bN}$. The reward function associated to an admissible strategy $\{h^n\}_{n \in \bN}$ is given by
\begin{gather*}
J\left(t,x,\{h^n\}\right)=\mathbb E_{(t,x)}^{\{h^n\}}\left[ U(W_T)\right],
\end{gather*}
and the value function of the optimization problem under partial information  is
\begin{gather}\label{1.7}
V(t,x) =V(\widetilde x)=\sup \left\{ J\left(t,x,\{h^n\} \right)\colon   \{h^n\}_{n \in \bN}\  \text{admissible strategy} \right \}\,.
\end{gather}

\subsubsection{The corresponding Markov decision model}\label{sec:MDM}
The optimization problem in \eqref{1.7} can be reduced to an optimization problem in an infinite horizon Markov decision model (MDM). Here we use the same techniques as in \cite{colaneri2016shall}, to solve the utility maximization problem from terminal wealth. To give an idea, we show that the value function of the piecewise deterministic control problem can be identified as the value function of a certain Markov decision problem that can be solved by a fixed point argument, see \cite[Chapter 8]{bauerle2011markov} for details.

 Although the main technique used to handle the optimization problem is similar, we briefly explain the differences with \cite{colaneri2016shall}. In \cite{colaneri2016shall}, the authors study an optimal liquidation problem for an investor whose actions directly affect the stock price dynamics by increasing the intensity of downward jumps in a partial information setting. The stock price dynamics is given by a pure-jump process. The goal is to maximize the expected total reward represented by a functional consisting in a combination of running profits, which linearly depend on the liquidation rate (that is the control), and terminal value representing the price of a block transaction at the final time. The first difference with our setup is the model: here has an indirect effect through the generator of the unobservable Markov chain. Moreover we have a different objective, as we aim to maximize the expected utility from terminal wealth. On the other hand, since the jump intensity of the PDMP is stochastic, unfortunately we cannot directly rely on the results in \cite{bauerle2010optimal,bauerle2011markov}. Finally, mimicking the argument in \cite{colaneri2016shall}, we are able to give a characterization of the optimal value function as the unique viscosity solution of the HJB, which permits a numerical study, while in \cite{bauerle2011markov} optimal strategies and optimal value functions are obtained by a policy iteration procedure, which has a fast convergence rate.

The infinite horizon Markov decision model corresponding to the PDMP can be introduced as follows.
We consider the sequence $\{L_n\}_{n \in \bN}$ of random variables defined by
\begin{gather*}
L_n=(T_n, X_{T_n})=\widetilde X_{T_n} \text{ for } \ T_n<T, \quad n\in\bN\, ,
\end{gather*}
and $L_n=\Delta$ for $T_n\ge T$ where $\Delta$ is some cemetery state. In other words a state $(t,x)=(t,w,\bpi)$ represents a jump time $t$ and the wealth $w$ and filter $\bpi$ just after the jump.

For a function $\alpha\in\mathcal A$ we denote by $\widetilde \varphi^{\alpha}_t(\widetilde x)$  the flow of the initial value problem
$\frac{\ud }{\ud s}  \widetilde X(s)=g \big(\widetilde X(s), \alpha_s\big)$ with initial condition  $ \widetilde X(0) =\widetilde x.$ Equivalently the piecewise deterministic process $\widetilde X$ is given by $\widetilde X_t=\widetilde \varphi^{\alpha}_{t-T_n}(\widetilde X_{T_n})$, for every $t \in [T_n, T_{n+1})$ before time $T$.
To stress dependence on time also use the notation $\widetilde \varphi_t^{\alpha}=(t,\varphi^{\alpha})$.

We define the functions
   \begin{align}\label{eq:def-survival}
   \lambda^{\alpha}_s(\widetilde x) & =\lambda(\widetilde \varphi^{\alpha}_s(\widetilde x),\alpha_s):=\lambda((t+s,\varphi^{\alpha}_s),\alpha_s),\\
    \Lambda^{\alpha}_s(\widetilde x) &=\Lambda^{\alpha}(s ; \widetilde x):= \!\!\int_0^s \lambda^{\alpha}_u(\widetilde x)\ud u.\nonumber
\end{align}


Now we want to introduce the transition kernel $Q_L$ of the Markov decision model $\{L_n\}_{n\in \bN}$. The distribution of the interarrival times $T_{n+1}-T_n$ given $L_{n}=(t,x)$ and $h^n=\alpha$ is equal to $\lambda^\alpha(\widetilde{x})e^{-\Lambda^\alpha_u(\widetilde x)}\ud u$, where $\widetilde x=(t,x)$.
Then for any bounded measurable function $f:\widetilde X \cup \{\Delta\}\to \R$, the transition kernel of the MDM is given by
\begin{equation*}\label{eq:form-of-QL}
Q_L f \big( (t,x), \alpha\big ) =  \int_0^{ T-t } \lambda^\alpha_u(\widetilde x) e^{-\Lambda^\alpha_u(\widetilde x)} Q_{\widetilde X} f(u+t, \varphi_u(\widetilde x), \alpha_u \big ) \ud u + e^{- \Lambda_{\tau^\varphi}^\alpha (\tilde x)} f(\bar \Delta),
\end{equation*}
with $Q_L \I_{\{ \Delta\}} ( \Delta, \alpha)=1$.


To define the one-stage reward function  $r \colon \widetilde{\mathcal X}\times \mathcal A\to \mathbb{R}_{\ge 0}$, we first indicate by $w_t$, the wealth component of the flow $\widetilde \varphi^{\alpha}$. Then we have that
\begin{align*}
r(\widetilde x, \alpha)=e^{-\Lambda^{\alpha}_{T-t}(\widetilde x)} U(w_{T-t}), \quad r(\Delta) =0.
\end{align*}
The expected reward of a policy $\{h^n\}_{n \in \N}$ is given by
$$J_{\infty}^{\{h^n\}}(\widetilde x)=\mathbb{E}_{\widetilde x}^{\{h^n\}}\left[ \sum_{n=0}^{\infty}r\left( L_n,h^n(L^n)\right) \right],$$ and
\begin{align}
J_{\infty}(\widetilde x):=\sup \left\{J_{\infty}^{\{h^n\}}(\widetilde x):\{h^n\} \ \bF^S-\text{admissible strategy}   \right\}.\label{2.8}
\end{align}

Now we need to verify that this construction of an infinite-stage Markov decision model leads to an optimal control problem which is equivalent to the original PDP control problem. In  the next lemma we show that the value functions corresponding to the MDM and the control problem for PDMP coincide. The proof is provided in Appendix \ref{app:proofs}.

\begin{lemma}\label{lemma2.1} It holds for all $\bF^S$admissible strategies $\{h^n\}_{n\in \bN}$ that $V^{\{h^n\}}=J_{\infty}^{\{h^n\}}$ and hence $V=J_{\infty}$, that is, control problems \eqref{1.7} and \eqref{2.8} are equivalent.
\end{lemma}
Define the operator $\mathcal T$ of the Markov decision model as
\begin{equation*} \label{eq:optimality-equation-explicit1}
\mathcal T v(\widetilde x) := \sup_{\alpha\in\mathcal{A}}\Big \{e^{-\Lambda^{\alpha}_{T-t}(\widetilde x)} U(w_{T-t})+  \int_0^{T-t}  \lambda^\alpha_u(\widetilde x) e^{-\Lambda^\alpha_u(\widetilde x)} Q_{\widetilde X} v(t+u, \varphi_u(\widetilde x), \alpha_u \big ) \ud u\Big\}.
\end{equation*}

Our idea is to characterize the value function as the unique fixed point of the operator $\mathcal T$. To this  we need to prove that continuity for the reward function and the transition kernel hold over a class of admissible controls which is compact. Therefore, according to the general theory, we enlarge the action space introducing the set of relaxed controls, and define a suitable topology on this space, called the {\em Young Topology}. We refer to \cite{davis1993markov,bauerle2011markov} for more details.

The set of \emph{relaxed controls} is given by
\begin{gather*}\widetilde{\mathcal{A}}:=\{\alpha:[0,T]\to  \mathcal{M}^1([-L,L])\  \},
\end{gather*}
where $\mathcal{M}^1([-L,L])$ is the set of probability measures on $[-L,L]$.

In the context of relaxed control, we define an admissible relaxed strategy as a sequence of mappings $\{\nu^n\}:\widetilde{\mathcal X}\to \widetilde{\mathcal A}$.

To make the set $\widetilde{\mathcal{A}}$ compact, we introduce the Young topology as the coarsest topology such that all mappings of the form
\[\alpha\to \int_0^{T}\int_{-L}^{L}f(t,u)\alpha_t(\ud u)\ud t\]
are continuous for all functions $f:[0,T]\times[-L,L]\to\mathbb{R}$ that are continuous in the second argument, measurable in the first one and  $\int_0^T  \max_{u \in [-L,L]} |f(t,u)| \ud t <\infty$ (see, e.g \cite[Chapter 8]{bauerle2011markov}).

We remark that, as pointed out in \cite{bauerle2011markov, colaneri2016shall}, non-relaxed control form a dense subspace of relaxed controls.

For a measurable function $v:[-L,L]\to \R$ and some measure $\xi \in  \mathcal{M}^1([-L,L])$, we define
$ \langle \xi, v \rangle := \int_{-L}^L v (\nu) \xi(\ud \nu)$.
In order to use the properties of the set $\widetilde{\mathcal A}$ we now extend some definitions for  $\alpha \in \widetilde{ \mathcal{A}}$. First,  the vector fields $g$ of the PDMP becomes
\begin{align*}
g(\widetilde x, \alpha)= \langle \alpha,  g(\widetilde x, \cdot) \rangle = \int_{-L}^L g(\widetilde x, \nu)\alpha_s(\ud \nu),
\end{align*}
the jump intensity is given by
$\lambda^{\alpha}_s(\widetilde x)= \langle \alpha_s (\ud \nu),  \lambda(t+s,\varphi^{\alpha}_s, \nu)\rangle,$ and $ \Lambda^{\alpha}_s = \Lambda^{\alpha}_s(\widetilde{x})=\int_0^{s}\lambda^{\alpha}_u (\widetilde{x})\ud u,$ the reward function
\begin{align*} \label{eq:r-extended}
r(\widetilde x,\alpha) = e^{-\Lambda^{\alpha}_{T-t}}U(w_{T-t})\,,\quad
\end{align*}
and finally the transition kernel is
\begin{equation*}\label{eq:QL-extended}
Q_L v\big(\widetilde{x},\alpha\big)  = \int_0^{T-t} \!\! \lambda^\alpha_u(\widetilde x) e^{- \Lambda_u^\alpha }  \langle \alpha_u (\ud \nu)  , Q_{\widetilde X} v( t+u, {\varphi}_u(\widetilde x), \nu  \big ) \rangle  \ud u + e^{- \Lambda_{T-t}^\alpha } v(\bar \Delta) \, .
\end{equation*}

Moreover we have the following extension of the operator $\mathcal T$:

\begin{align*}
\mathcal T \phi(\widetilde x) =  \sup_{\alpha\in\widetilde{\mathcal{A}}} \left(r(\widetilde x,\alpha)+Q_L \phi \big(\widetilde{x},\alpha\big) \right).
\end{align*}

In the next lemma we show that there exists a bounding function for the MDM and the MDM is contracting. This is essential to prove that the value function is the unique fixed point of the operator $\mathcal{T}$.

\begin{definition}
A function $b:\widetilde{\mathcal{X}}\to\mathbb{R}_{\ge 0}$ is called a \emph{bounding function} for a MDM, if there are constants $c_r,c_b >0 $ such that
$ |r(\widetilde x,\alpha)| \le c_r b(\widetilde x)$  and $ Q_L b(\widetilde x, \alpha) \le c_b b(\widetilde x)$ for all $(\widetilde x,\alpha)\in \widetilde{\mathcal{X}}\times\mathcal{A}$.
If  moreover $c_b<1$, the MDM is \emph{contracting}.
\end{definition}

We define for a bounding function $b$ the set $\mathcal B_b$ of functions $v:\widetilde{\mathcal{X}}\to \mathbb{R}$ such that $v(\widetilde x)\le C b(\widetilde x)$.

\begin{lemma}\label{lemma3.3}
$b(\widetilde{x})=b(t,x)=e^{c(T-t)}s, \quad c\ge 0$, and $b(\bar{\Delta})=0$, is a bounding function and the MDM with the kernel $Q_L$ is contracting for sufficiently large $c$.
\end{lemma}
The proof of the lemma is postponed to Appendix \ref{app:proofs}.

We make now an assumption that provides continuity conditions on the data of our model. 

\begin{ass}\label{ass3}

For any sequence $\{(t_n, \pi_n)\}_{n \in \bN}$, with $ (t_n, \bpi_n) \in [0,T) \times  \Delta^K$, such that $\ds (t_n,  \bpi_n) \xrightarrow[n \to \infty]{} (t, \bpi)$, the functions $u^i(t,\bpi,z)$ given in Proposition \ref{prop:KS} satisfy
$$\lim_{n \to \infty} \ \sup_{z \in \text{supp}(\eta^\P)} |u^i(t_n, \bpi_n,z) - u^i(t,  \bpi, z) | =0\,,$$
where $\text{supp}(\eta^\P)$ indicates the set $\{z \in \R: \ \eta^P(t, e_i, z)\neq 0, \ t\in [0,T], \ i\in\{1, \dots, K\}\}$.
\end{ass}

Then the following result holds.

\begin{proposition}\label{prop:continuity} Under Assumption \ref{ass3}, the mappings $(\widetilde{x},\alpha) \mapsto  r(\widetilde x, \alpha)$ and
$(\widetilde{x},\alpha) \mapsto Q_L v (\widetilde{x},\alpha)$ for every $v \in \mathcal B_b$, are continuous on $\widetilde{ \mathcal{X}} \times \widetilde{\mathcal{A}}$ with respect to the  Young topology on
$\widetilde{\mathcal{A}}$.
\end{proposition}
The proof is provided in Appendix \ref{app:proofs}.

The main result of the section concerns with the existence and uniqueness of the solution of the corresponding fixed point equation and it is summarized in the next theorem.

\begin{theorem}\label{prop:fixed_point} Suppose that Assumption \ref{ass3}  holds.  Then we have:
\begin{itemize}
\item[i)] the value function $V$ is continuous on $\widetilde{\mathcal X}$ and satisfies the boundary conditions  $V(T, w, \bpi)= U(w)$.
\item[ii)] $V$ is the unique fixed point of the operator ${\mathcal{T}}$ in $\mathcal B_b$. 
    \end{itemize}
\end{theorem}

\begin{proof}
First note that by Lemma~\ref{lemma3.3} a bounding function for our model is given by $b(t,w,\bpi)=e^{\gamma (T-t)} w,$ for some $\gamma >0$ and that the MDM is contracting.
Following the Proposition \ref{prop:continuity}, we get that the reward function $r$ and transition kernel $Q_L$  are continuous with respect to the Young topology.
By applying \cite[Theorem 7.3.6]{bauerle2011markov}  we obtain that $V$ is the fixed point of the maximal reward operator extended to the class of the relaxed controls and finally the result of the theorem follows from \cite[Corollary 4.10]{colaneri2016shall}.
\end{proof}

In order to provide a characterization of the optimal value function in terms of the solution of a suitable HJB equation, we resort to the viscosity solution analysis. This also legitimates the numerical study which will be done in the next section.

As a first step we want to reduce the problem to the case where the state process takes values in a compact set. Since the case of logarithmic utility is a limiting case of the power utility we only write the reduction for the latter. By using positive homogeneity we have that
\begin{equation*}
V(t,w,\bpi)=\frac{w^\theta}{\theta} \overline{V}(t, \bpi).
\end{equation*}
Define the compact set $\widetilde{\mathcal Y}:=[0,T]\times \Delta_K$.

We now define $\overline{g} \colon \widetilde{\mathcal{Y}} \times [-L, L]\to \R^{K+2}$ by identifying $$(\overline{g})^{(1)} = g^{(1)},  \, \text{ and } (\overline{g})^{(k+1)} = g^{(k+2)}\,,\; k=1,\dots,  K\,.
$$
and denote by $ \overline{\varphi}_u (\alpha, \widetilde y) $ the flow of $\overline{g}$.

Since the jump intensity $\lambda$ introduced in \eqref{eq:def-survival} is independent of $w$, by Theorem \ref{prop:fixed_point}, the optimality equation for $\overline{V}$ is given by
\begin{align*}
\overline{V} (\widetilde y) = \sup_{\alpha \in A}  \Big \{
\int_0^{T-t }
 \lambda^\alpha_u(\widetilde y) e^{-\Lambda^\alpha_u(\widetilde y)} \overline{Q} \  \overline{V}  \big(u+t, \overline{\varphi}_u (\alpha,\widetilde y)  , \alpha_u  \big) \ud u   + \frac{1}{\theta} e^{-   \Lambda_{T-t}^\alpha (\widetilde y)}  \Big \}, \label{eq:optimality-for-Vprime}
\end{align*}
where, for $h \in [-L,L]$, $\widetilde y  \in \widetilde{\mathcal{Y}}$, and any measurable function $\Psi \colon \widetilde{\mathcal{Y}} \to \R_{\ge 0}$, $\overline{Q}$ defines the new transition kernel
$$
\overline{Q} \Psi ( \widetilde y , h) := \lambda(\widetilde y) \sum_{j=1}^K \pi^j \int_{\mathbb{R}} (1+hz)^\theta \Psi \left(t,(\pi^i(1+u^i(t,\pi,z)))_{i=1,\dots,K} \right) \eta^{\P}(t,e_j,\ud z).
$$

This, in turn, implies that the value function $\overline{V}$ satisfies  $\overline{V}  = \overline{\mathcal{T}} \overline{V}$, with the reward operator $\overline{\mathcal{T}}$ given by
\begin{align*}
\overline{\mathcal{T}} \Psi (\widetilde y) &= \sup_{\alpha \in A}  \Big \{
\int_0^{T-t }
 \lambda^\alpha_u(\widetilde y) e^{-\Lambda^\alpha_u(\widetilde y)} \overline{ Q}   \Psi  \big(u+t, \overline{\varphi}_u (\alpha,\widetilde y)  , \alpha_u  \big) \ud u   + \frac{1}{\theta} e^{-   \Lambda_{T-t}^\alpha (\widetilde y)}  \Big \}.
\end{align*}
In the sequel we aim to show that  $\overline{V}$  solves, in the viscosity sense, the equation
\begin{equation} \label{eq:HJB-general}
F_{\overline{V}} \big ( \widetilde y, \overline{V}(\widetilde y), \nabla \overline{V} (\widetilde y) \big) = 0, \text{ for } \widetilde y \in  \widetilde{\mathcal{Y}}^0, \quad \overline{V}(\widetilde y) = \frac{1}{\theta} \text{ for } \widetilde y \in \partial  \widetilde{\mathcal{Y}},
\end{equation}
where, for $\Psi \colon \widetilde{\mathcal Y} \to \R_{\ge 0}$,  the function $F_\Psi \colon \widetilde{\mathcal Y} \times \R_{>0} \times \R^{K+1} \to \R $ if given by
$$ F_\Psi ( \widetilde y, v, p) = -\hspace{-0.1cm}\sup_{\nu \in [-L, L]} \big \{- \lambda( \widetilde{y}, \nu) v  +  \overline{g}(\widetilde{y}, \nu) p + \overline{Q}\Psi (\widetilde y, \nu) \big \}\,.$$

The following result, proven in \cite[Theorem 5.3]{colaneri2016shall} applies.

\begin{theorem}\label{thm:viscosity} Suppose that  
the Markov chain $Y$ has no absorbing state ($-q_{ii} > 0 $ for all $i\in \{1, \dots, K\}$). Then the value function $\overline{V}$ is the unique continuous viscosity solution of \eqref{eq:HJB-general} in $\widetilde{\mathcal{Y}}$ and a comparison principle holds.
\end{theorem}

In more explicit terms the HJB equation for the value function in the partial information setting, can be written as


\begin{align}
\nonumber&0=\!\! \sup_{h \in [-L,L]} \Bigg\{ \frac{\partial V}{\partial t}(t,w,\bpi)+w(1-h)\rho \frac{\partial V}{\partial w}(t,w,\bpi)\\
\nonumber&+\sum_{k,j=1}^K\frac{\partial V}{\partial \pi^k}(t,w,\bpi)\pi^j \left( q^{jk}(h)-\int_{\mathbb{R}} \pi^k u^k(t,z) \eta^\P(t,e_j,\ud z)\right)\\
\label{eq:HJB_2}&+\sum_{j=1}^K \pi^j \!\!\int_\R \! \left[V\left(t,w(1+h z), (\pi^{i}(1+u^{i}(t,z)))_{i\in \{1,...,K\}}\right)-V(t,w,\bpi)\right] \eta^\P(t,e_j,\ud z)       \Bigg\}.
\end{align}

In the next sections we analyze the case of logarithmic and power utility functions in detail, in the partial information framework.
\subsection{Logarithmic utility under partial information}\label{sec:log_utility_partial_info}
According to analysis conducted in the full information framework, we study the optimization problem with and without impact. For the logarithmic utility preferences this leads to two different approaches: in the first case pointwise maximization applies, while in the second one we need to use a dynamic programming approach. We also provide a comparison between the optimal strategies under full and partial information.
\subsubsection{Logarithmic utility - No market impact}
We first assume that the investor has no impact, meaning that entries in the generator of the Markov chain do not depeend on the trading strategy, and we solve the optimal control problem directly. For a fixed strategy $h\in \widetilde {\mathcal H}$ by applying the It\^{o} formula we get
\[
V(t,w,\bpi)=\log(w)+\sup_{h \in \widetilde{\mathcal H}} \widetilde{B}(t,\bpi;h),
\]
where
\begin{align*}
\widetilde{B}(t,\bpi;h)=\bE^{t,\bpi}\left[\int_t^T\left((1-h_s)\rho+\sum_{i=1}^K\pi^i_s\int_{\R}\log(1+h_s z) \eta^{\P}(s,e_i,\ud z)\right)\ud s \right. \\
\left.+\int_t^T\int_{\mathbb{R}} \log(1+h_{s} z){\nu}^\pi(\ud s, \ud z)\right],
\end{align*}
and proposition below holds.



\begin{proposition}\label{logu3} Suppose $U(w)=\log(w)$ for $w>0$.
\begin{itemize}
\item[i)] Let $h^*(t, \bpi)$ 
satisfy either
\begin{align}\label{eq:logoptimal}
\sum_{j=1}^K \pi^j_t\int_{\R}\frac{z}{1+h^{\ast}(t,\bpi) z} \eta^\P(t,e_j,\ud z)=\rho
\end{align}
or $h^{\ast}(t,\bpi)\in\{-L,L\}$. Then the optimal strategy $h^*_t=h^*(t, \bpi)$ for every $t \in [0,T]$ and $\bpi \in \Delta_K$.
\item[ii)] The value function is of the form
\end{itemize}
\small {\begin{equation*}\label{eq:log_utility_partial}
V(t,w ,\bpi)=\log(w)+\bE^{t,\bpi}\left[\int_t^T\left((1-h^{{\ast}}_s)\rho+\sum_{i=1}^K\pi^j_s \int_{R}\log(1+h^{{\ast}}_{s} z) \eta^{\P}(s,e_i,\ud z)\right)\ud s \right].
\end{equation*}}
\end{proposition}

\begin{proof} The proof follows from the same arguments of that of Lemma~\ref{logu2}.
\end{proof}

\begin{remark}
Comparing the results in Proposition \ref{logu2} and Proposition \ref{logu3} we observe similar structures for the optimal strategies. Precisely in the partial information case the optimal strategy solves an equation of the form \eqref{eq:logoptimal} where $(\bF, \P)$-compensator of the jump measure is replaced by the $(\bF^S, \P)$-compensator. Intuitively this is due to the myopic property of the logarithmic utility; the agent replaces the unobserved local characteristics of the return process by their filtered estimates ignoring the extra risk associated with the information uncertainty (see, for example, \cite{feldman1992logarithmic}).
\end{remark}

\subsubsection{Logarithmic utility - Market impact}

According to full information, when there is an impact on the state of the Markov chain we cannot apply pointwise maximization, but we can characterize the value function as the solution of the HJB equation, in the viscosity sense. Here we propose the following ansatz $V(t,w,\bpi)=\log (w) + B(t,\bpi)$, for some function $B$ with the terminal condition $B(T,\bpi)=0$, for all $\bpi\in\Delta_K$.
Substituting this form of the value function into \eqref{eq:HJB_2}, we obtain the following equation:

{\small \begin{align}
 0&=  \sup_{h \in [-L,L]} \Bigg\{ \frac{\partial B}{\partial t}(t,\bpi)+ (1-h)\rho \\
\nonumber&+\sum_{k,j=1}^K\frac{\partial B}{\partial \pi^k}(t,\bpi)\pi^j \left( q^{jk}(h)-\int_{\mathbb{R}} \pi^k u^k(t,z) \eta^\P(t,e_j,\ud z)\right)\\
\label{eq:logpartial1} & +\sum_{j=1}^K \pi^j\int_\R \log(1+hz) + \left[B\left(t, (\pi^{i}(1+u^{i}(t,z)))_{i\in \{1,...,K\}}\right)-B(t,\bpi)\right] \eta^\P(t,e_j,\ud z)      \Bigg\}.
\end{align}}


By Theorem \ref{thm:viscosity}, the value function is the unique viscosity solution of  problem  \eqref{eq:logpartial1}. 
Given the form of the compensator $\eta^\P$, the equation can be solved, for instance using a numerical scheme.

\subsection{Power utility under partial information}\label{sec:power_utility_partial_info}
In this part we will work under the assumption of power utility, that is, $U(w)=\frac{1}{\theta}w^{\theta}$, $\theta<1$, $\theta \neq 0$.
Then the value function of the investor is
\begin{align*}V(t,w,\bpi)=\underset{h\in \widetilde{\mathcal H}}\sup \mathbb{E}^{t,w,\bpi}\left[\frac{1}{\theta}(W_T)^{\theta}\right],\end{align*}
where $\mathbb{E}^{t,w,\bpi}[\cdot]$ denotes the conditional expectation given $W_t=w$ and $\pi_t=\bpi$.

We know that, by positive homogeneity, the value function can be rewritten as
$V(t,w,\bpi)=\frac{1}{\theta}w^{\theta}\Gamma(t,\bpi)$, for some function $\Gamma:[0,T]\times \Delta_K\to \R_{>0}$ with $\Gamma(T,\bpi)=1$, for all $\bpi\in\Delta_K$.
Substituting this form of the value function into \eqref{eq:HJB_2}, we obtain the equation:
\begin{align*}
 0&=  \sup_{h \in [-L,L]} \Bigg\{ \frac{\partial \Gamma}{\partial t}(t,\bpi)+ \Gamma(t,\bpi)\theta(1-h)\rho \\
\nonumber&+\sum_{k,j=1}^K\frac{\partial \Gamma}{\partial \pi^k}(t,\bpi)\pi^j \left( q^{jk}(h)-\int_{\mathbb{R}} \pi^k u^k(t,z) \eta^\P(t,e_j,\ud z)\right)\\
\label{geqn2} & +\sum_{j=1}^K \pi^j\int_\R \left[(1+hz)^{\theta}\Gamma\left(t, (\pi^{i}(1+u^{i}(t,z)))_{i\in \{1,...,K\}}\right)-\Gamma(t,\bpi)\right] \eta^\P(t,e_j,\ud z) \Bigg\}.
\end{align*}



Therefore, after reduction, we deal with a problem having a bounded state space $[0,T]\times \Delta_K$. Theorem \ref{thm:viscosity} ensures existence and uniqueness of a viscosity solution for this problem. We solve it numerically in case of a two-state Markov chain in the next section.

\section{A Model with a Two-State Markov Chain}\label{sec:toy_model}
Suppose that we have a state process $Y$ described by a Markov chain with the state space $\mathcal{E}=\{e_1,e_2\}$. Without loss of generality we may assume that $e_1$ represents the good (bull) state of the market and $e_2$ is representing a bad (bear) state. We consider the situation where the investor's assets holdings are taken as a signal for the rest of the market that tends to behave accordingly. Then for the market, intensities of switching between the bull and the bear state depend on the portfolio weights of the reference ``large'' investor. In the current setting, we also assume that the impact of the portfolio choices on the Markov chain is linear, and assume that the infinitesimal generator has the form
\begin{gather*}
q^{12}(h_t)=a_1-b_1h_t, \quad q^{21}(h_t)=a_2+b_2h_t.
\end{gather*}
To guarantee that the entries $q^{1,2}$ and $q^{2,1}$ of the matrix stay positive we take $a_1, a_2>0$, $b_1\in (0, a_1/L)$ and $b_2\in (0, a_2/L)$.

This choice for the generator has the following motivation. If the investor buys, then she tends to increases the probability for the market to stay in (resp. switch to) the bull state, provided that the current state of the market is bull (resp. bear). Conversely, when the investor sells, the probability to stay in or jump to the bear state increases. This mechanism reflects certain real world situations such as manipulation and herding, which are frequently observed in markets where large investors are involved.

We assume that the return process may have two possible jump sizes, $\Delta R\in\{-\vartheta, +\vartheta\}$. Formally, it is given by
\begin{equation*}\label{ex1rdyn}R_t:=N^-_t+N^+_t, \quad t \in [0,T],\end{equation*}
where
\[\ud N^-_t= \int_{\R} -\vartheta\I_{[-\lambda^{-}(Y_{t^-}), 0]}(\zeta)\N(\ud t, \ud \zeta), \quad \ud N^+_t= \int_{\R} \vartheta\I_{[0,\lambda^{+}(Y_{t^-})]}(\zeta)\N(\ud t, \ud \zeta)\]
are two Poisson processes with jump sizes $\vartheta$ and intensities $\lambda^+(e_i)=\lambda^+_i, \lambda^-(e_i) = \lambda^-_i$, $i\in\{1,2\}$,  for some constants $\lambda^+_1, \lambda^+_2, \lambda^-_1, \lambda^-_2>0$ and such that $\lambda^+_1>\max\{\lambda^-_1, \lambda^+_2\}$ and $\lambda^-_2> \max\{\lambda^-_1, \lambda^+_2\}$.
This conditions imply that the intensity of an upward jump is larger in the bull state of the market compared  the bear one.  Moreover in the bull state it is more likely to observe an upward jump then a downward jump.
In this example we take the Poisson random measure $\N(\ud t, \ud \zeta)$ with intensity $\varsigma(\ud \zeta)\ud t=\I_{[-\lambda^-_2,\lambda^+_1 ]}\ud \zeta \ud t$. Then the compensator has the form
$$\eta^\P(t, e_i,\ud z) = \lambda^+_i \delta_{\{\vartheta\}}(\ud z) +
 \lambda^-_i  \delta_{ \{-\vartheta \}}(\ud z),
$$
where $\delta_{ \{x \}}(\ud z)$ is the Dirac mass at point $x$.
Notice that here Assumption \ref{ass2} is satisfied for $-\frac{1}{\vartheta}<h_t<\frac{1}{\vartheta}$.

In the reminder of this section we are going to compare the results for logarithmic and power utility choices for the full and the partial information settings.
\subsection{Logarithmic utility}
First we consider a investor with full information on the market. Starting with the case of no market impact, applying Proposition \ref{logu2} and checking the first and second order conditions we obtain the optimal strategy
\begin{equation*}\label{eq:log_full}
  h^{*}(t, e_i)=\frac{\lambda^+_i+\lambda^-_i}{2\rho}-\sqrt{\left(\frac{\lambda^+_i+\lambda^-_i}{2\rho}+\frac{1}{\vartheta}\right)^2-2\frac{\lambda^+_i}{\vartheta \rho}}.
\end{equation*}

In particular, if $\lambda^+_i=\lambda^-_i=\lambda_i$ we get that $h^*(t, e_i)=\frac{\lambda_i}{\rho}-\sqrt{\frac{\lambda^2_i}{\rho^2}+\frac{1}{\vartheta^2}}$. Finally we can characterize the value function as
\begin{align*}
V(t,w ,e_i)=&\log(w)+\bE^{t,i}\!\left[\int_t^T\!\!\!\!\left((1-h^{\ast}(s,e_1))\I_{\{Y_s=e_1\}}\rho+(1-h^{\ast}(s,e_2))\I_{\{Y_s=e_2\}}\rho\right)\ud s\right.\\
&+\int_t^T\int_{R}\log(1+h^{\ast}(s,e_1) z) \I_{\{Y_{s^-}=e_1\}} \eta^{\P}(\ud s,e_1,\ud z) \\
&+\left.\int_t^T\int_{R}\log(1+h^{\ast}(s,e_2) z) \I_{\{Y_{s^-}=e_2\}} \eta^{\P}(\ud s,e_2,\ud z) \right].
\end{align*}

For the case where the impact is non-zero, the value function can be characterized as $V(t,w,e_i)=\log(w)+\beta(t,e_i)$, $i \in \{1,2\}$ with the functions $\beta(t,e_1)$ and $\beta(t,e_2)$ solving

\begin{align}
\frac{\ud \beta}{\ud t}(t,e_1)=& - \sup_{h \in [-L,L]} \Bigg\{    (1-h) \rho  + \left(\beta(t,e_2)-\beta(t,e_1)\right) (a_1-b_1 h) \nonumber\\
&+\int_\R \log(1+h z) \eta^{\P}(t,e_1,\ud z)\Bigg\},\nonumber \\
\frac{\ud \beta}{\ud t}(t,e_2)=& - \sup_{h \in [-L,L]} \Bigg\{    (1-h) \rho  +\left(\beta(t,e_1)-\beta(t,e_2)\right) (a_2+b_2 h)\nonumber\\
&+\int_\R \log(1+h z) \eta^{\P}(t,e_2,\ud z)\Bigg\},\nonumber
\end{align}
respectively, with boundary conditions $\beta(T,e_i)=0$, for $i=\{1,2\}$.

Assume now that the investor's information is given by the filtration $\bF^S$. There, by Proposition \ref{logu3},  the optimal strategy in case of no market impact turns out to be
\begin{equation*}\label{eq:log_partialex}
  h^{*}(t,\bpi)=\frac{ \bpi^\top \Lambda^+ +\bpi^\top \Lambda^-}{2\rho}-\sqrt{\left(\frac{\bpi^\top \Lambda^+ +\bpi^\top \Lambda^-}{2\rho}+\frac{1}{\vartheta}\right)^2-2\frac{\bpi^\top \Lambda^+}{\vartheta \rho}}.
\end{equation*}
where $(\Lambda^+)^\top=(\lambda^+_1, \lambda^+_2)$ and similarly $(\Lambda^-)^\top=(\lambda^-_1, \lambda^-_2)$. This is the classical case where the optimal strategy has the same structure of that under full information in which the unobserved components are replaced by their filtered estimates. The stochastic representation of the value function is given by
\begin{align*}
V(t,w ,\bpi)=&\log(w)+\bE^{t,\pi}\left[\int_t^T(1-h^{*}(s,\pi_s))\rho\ud s\right.\\
&+\int_t^T\pi^1_s \int_{R}\log(1+h^{*}(s,\pi_s) z) \eta^{\P}(s,e_1,\ud z)\ud s \\
&\left.+ \int_t^T\pi^2_s \int_{R}\log(1+h^{*}(s,\pi_s) z) \eta^{\P}(s,e_2,\ud z) \ud s \right].
\end{align*}

In the partial information case, the value function has the form $V(t,w,\pi)=\log (w)+B'(t, \pi)$ where $B'(t,\pi)=B(t, \pi, (1-\pi))$ is the solution of the HJB equation

\begin{align*}
0=&\sup_{h \in [-L,L]}\left\{ \frac{\partial B'}{\partial t}(t, \pi)+ (1-h) \rho + \frac{\partial B'}{\partial \pi}(t, \pi)\left(\pi q^{11}(h)+(1-\pi) q^{21}(h)\right)\right.\\
&-\frac{\partial B'}{\partial \pi} \pi (1-\pi)(\lambda^+_1+\lambda^-_1-\lambda^+_2-\lambda^-_2)\\
& + (\pi\lambda^+_1+(1-\pi)\lambda^+_2)\log(1+hz) + (\pi\lambda^-_1+(1-\pi)\lambda^-_2)\log(1-hz)\\
 &  + \pi (\lambda^+_1+\lambda^-_1)\left[B'\left(t, \frac{\pi \lambda^+_1}{\pi \lambda^+_1+ (1-\pi) \lambda^+_2}\right)-B'(t,\pi)\right] \\
 &\left.+ (1-\pi)(\lambda^+_2+\lambda^-_2) \left[B'\left(t, \frac{\pi \lambda^+_1}{\pi \lambda^+_1+ (1-\pi) \lambda^+_2}\right)  -B'(t,\pi)\right] \right\}.
\end{align*}
An explicit solution of the above equation is difficult to find. In general, it is possible to apply numerical experiments to get the qualitative behavior of both the value function and the optimal strategy. Since the logarithmic utility case do not provide any simplification, for numerical study we only consider the power utility case.

\subsection{Power utility}
In the power utility case, when the investor has a full information on the state of the market, analysis of the optimization problem leads to solving the system
\begin{align*}
\frac{\ud \gamma}{\ud t}(t,e_1)=& - \sup_{h \in [-L,L]} \left\{    (1-h) \rho  +\frac{1}{\theta} \left(e^{\theta (\gamma(t,e_2)-\gamma(t,e_1))}-1\right) (a_1-b_1 h) \right.\nonumber\\
&\left.+\frac{\lambda^+_1}{\theta}\left((1+h \vartheta)^\theta-1\right) +\frac{\lambda^-_1}{\theta}\left((1-h \vartheta)^\theta-1\right)\right\},\\
\frac{\ud \gamma}{\ud t}(t,e_2)=& - \sup_{h \in [-L,L]} \left\{    (1-h) \rho  +\frac{1}{\theta} \left(e^{\theta (\gamma(t,e_1)-\gamma(t,e_2))}-1\right) (a_2+b_2 h)\right.\nonumber\\
&+\left.\frac{\lambda^+_2}{\theta}\left((1+h \vartheta)^\theta-1\right) +\frac{\lambda^-_2}{\theta}\left((1-h \vartheta )^\theta-1\right)\right\},
\end{align*}
with the final condition $\gamma(T,e_1)=\gamma(T,e_2)=0$.

For the solution we use the following algorithm. Let $(t_0, \dots, t_N)$ be the sequence of discretized time points with $t_0=0$ and $t_N=T$. Knowing the final conditions allows to compute easily the control $h^{\ast}_{T}$ at time $T$. Then, using a backward scheme we solve the corresponding ODE at $t_{N-1}$. Given the value at $t_{N-1}$, now we can compute the control $h^{\ast}_{t_{N-1}}$ and we proceed until $t_0=0$.

In the numerical analysis we use the set of parameters: $T=1$ year, $w=1$, $\rho=0$, $\vartheta=0.02$, $\lambda^+_1=10$, $\lambda^-_1=5$, $\lambda^+_2=5$, $\lambda^-_2=20$, $\theta=0.5$, $a_1=5$, $b_1=-0.1$, $a_2=5$, $b_2=0.1$.

In Figure \ref{fig:full info-opt str} we plot the optimal investment strategies for cases where the initial state of the Markov chain is bull (lighter line) or bear (darker line) both with (solid line) and without (dashed line) market influence. Firstly, we observe that in all cases the optimal strategies never reach the values $\{-L, L\}$ corresponding to $L=50$, meaning that there is always an interior solution. Secondly, we can see that as time approaches to maturity, the optimal strategy in the case with impact converges to the one in the no-impact case. Moreover, actions of the investor are very different when we compare cases with and without impact. Consider for instance the situation where the initial state is bull. We observe that in the no-impact case the strategy is constant and always positive, meaning that the investor always buys. On the other hand,  in the case with impact the investor short-sells if time to maturity is large. This kind of an action might be interpreted in the following way. The investor tries to produce a jump in the Markov chain and make advantage of lower prices that would prevail in a future time. Clearly, she switches her behavior as time to maturity becomes shorter, since there is not enough time to make such a change. For the case of initial bear state, we see that the investor always short-sells. This is reasonable for the current parameter choice as on average the prices tend to go down. For the case with impact, the strategy turns out to be more aggressive.

\begin{figure}[htbp]
\centering
\includegraphics[height=7cm]{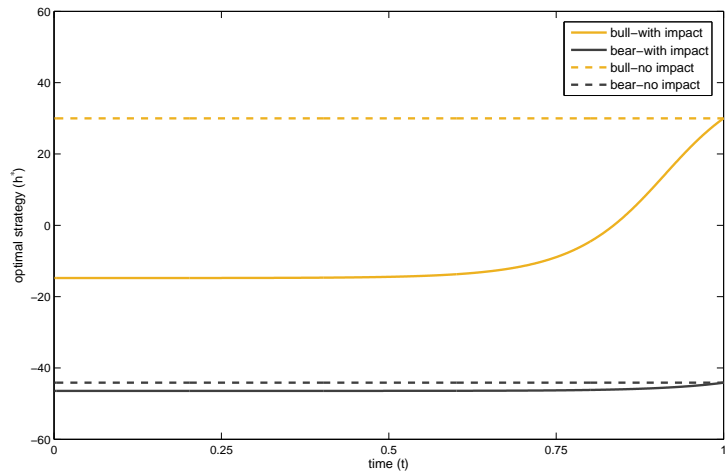}
\caption{Optimal strategy under full information  with (solid) and without (dashed) impact: $T=1$ year, $w=1$, $\rho=0$, $\vartheta=0.02$, $\lambda^+_1=10$, $\lambda^-_1=5$, $\lambda^+_2=5$, $\lambda^-_2=20$, $\theta=0.5$, $a_1=5$, $b_1=0.1$, $a_2=5$, $b_2=0.1$. }
\label{fig:full info-opt str}
\end{figure}

The different behavior for investors with market influence results in positive gains from utility maximization. Indeed, as we see in Figure \ref{fig:full info-value}, the value functions corresponding to the impact cases are sensibly larger than those corresponding to no-impact cases.  The optimal value corresponding to the bad state is larger than the optimal value for the initial good state. This is a consequence of the fact that the investor is allowed to short-sell, and clearly this also depends on our choice of the intensities of upward and downward jumps. In other words, we see that there is no absolute good and bad state.

\begin{figure}[htbp]
\centering
\includegraphics[height=7cm]{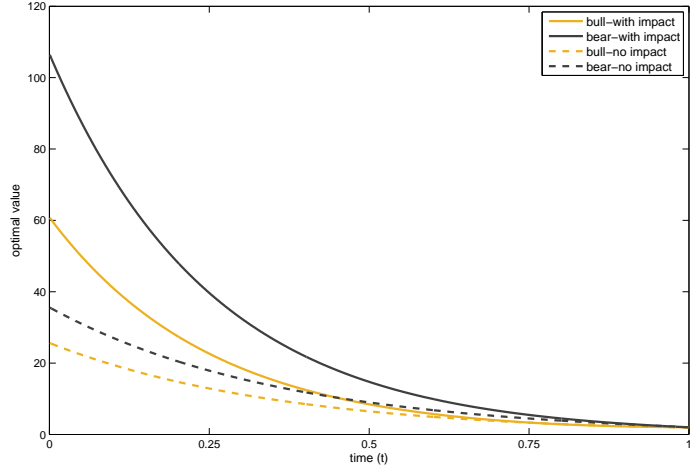}
\caption{Optimal value under full information with (solid) and without (dashed) impact: $T=1$ year, $w=1$, $\rho=0$, $\vartheta=0.02$, $\lambda^+_1=10$, $\lambda^-_1=5$, $\lambda^+_2=5$, $\lambda^-_2=20$, $\theta=0.5$, $a_1=5$, $b_1=0.1$, $a_2=5$, $b_2=0.1$. }
\label{fig:full info-value}
\end{figure}

Suppose now that the available information for the investor is given by $\bF^S$. Note that, for a two-state Markov chain we have $\pi^1+\pi^2=1$, then denote $\pi^1$ with $\pi$, we can define $\Gamma'(t, \bpi):=\Gamma(t, \pi, (1-\pi))$ and reduce the dimension of the optimization problem. In this case the function $\Gamma'$ can be characterized as the solution of the HJB

\begin{align*}
0=&\sup_{h \in [-L,L]}\left\{ \frac{\partial \Gamma'}{\partial t}(t, \pi)+ \Gamma'(t, \pi) \theta (1-h) \rho + \frac{\partial \Gamma'}{\partial \pi}(t, \pi)\left(\pi q^{11}(h)+(1-\pi) q^{21}(h)\right)\right.\\
&-\frac{\partial \Gamma'}{\partial \pi} \pi (1-\pi)(\lambda^+_1+\lambda^-_1-\lambda^+_2-\lambda^-_2)\\
&+(\pi \lambda^+_1+(1-\pi) \lambda^+_2)\left( (1+h\vartheta)^\theta \Gamma'\left(t, \frac{\pi \lambda^+_1}{\pi \lambda^+_1+ (1-\pi) \lambda^+_2}\right)-\Gamma'(t, \pi)\right)\\
&\left.+(\pi \lambda^-_1+(1-\pi) \lambda^-_2)\left( (1-h\vartheta)^\theta \Gamma'\left(t, \frac{\pi \lambda^-_1}{\pi \lambda^-_1+ (1-\pi) \lambda^-_2}\right)-\Gamma'(t, \pi)\right)\right\}.
\end{align*}

Since, in general it is not possible to find an explicit solution to the above maximization problem we deepen our analysis through numerical experiments. In the case of partial information, we use an explicit finite difference method to solve the corresponding partial integro-differential equation. In order to guarantee the positivity of the scheme we use forward-backward approximation for the first order derivatives (see, for instance, \cite{cont2005finite}). Also, to ensure the convergence of the scheme we verify the usual consistency and stability conditions.

As in the full information case we study both the optimal strategy and the value function, and obtain results that are consistent with those in the full information setting. We observe in Figure \ref{fig:partial info-opt str} that optimal strategies in the impact case converge, for values of time close to maturity, to optimal strategies in the no-impact cases.
Moreover the interesting behavior of the investor with an impact is preserved: for the initial bull state she short-sells when the time is far from maturity and in the initial bear state the strategy is always more aggressive.

\begin{figure}[htbp]
\centering
\includegraphics[height=7cm]{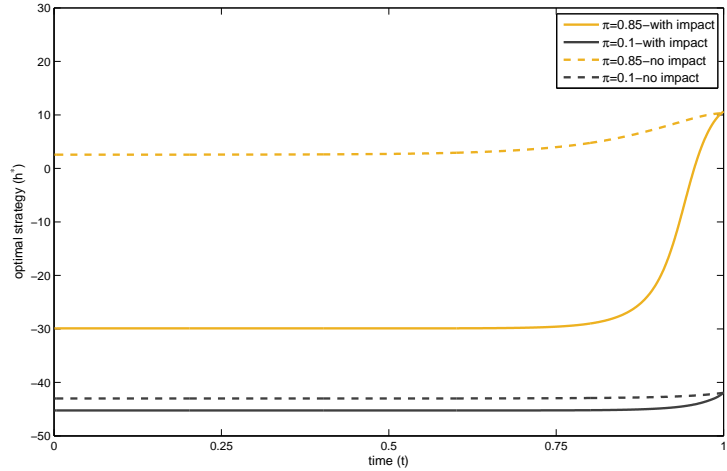}
\caption{Optimal strategy under partial information with (solid) and without (dashed) impact: $T=1$ year, $w=1$, $\rho=0$, $\vartheta=0.02$, $\lambda^+_1=10$, $\lambda^-_1=5$, $\lambda^+_2=5$, $\lambda^-_2=20$, $\theta=0.5$, $a_1=5$, $b_1=0.1$, $a_2=5$, $b_2=0.1$.} 
\label{fig:partial info-opt str}
\end{figure}

The value function is consistently larger for the investor with an impact, see Figure \ref{fig:partial info-value}.

\begin{figure}[htbp]
\centering
\includegraphics[height=7cm]{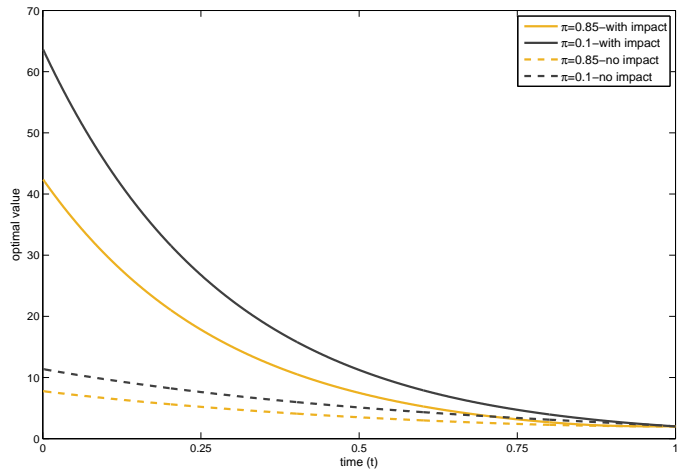}
\caption{Optimal value under partial information with (solid) and without (dashed) impact: $T=1$ year, $w=1$, $\rho=0$, $\vartheta=0.02$, $\lambda^+_1=10$, $\lambda^-_1=5$, $\lambda^+_2=5$, $\lambda^-_2=20$, $\theta=0.5$, $a_1=5$, $b_1=0.1$, $a_2=5$, $b_2=0.1$. }
\label{fig:partial info-value}
\end{figure}

Finally, we analyze the gains from filtering. In order to do that we compare the value functions corresponding to two investors. The first one uses the optimal strategy obtained in the partial information setting, while the second one ignores the presence of two different regimes in the market. Instead the second one uses the average parameters,
${\lambda^+}=\lambda^+_1 p+\lambda^+_2 (1-p), \quad {\lambda^-}=\lambda^-_1 p+\lambda^-_2 (1-p)$, where $p=\frac{a_2}{a_1+a_2}$.  In Figure \ref{fig:gains}, we observe that the investor's gains from using filtered estimates, instead of the average parameters, are always non-negative. Those profits justify the additional complexity induced by partial information.
\begin{figure}[htbp]
\centering
\includegraphics[height=7cm]{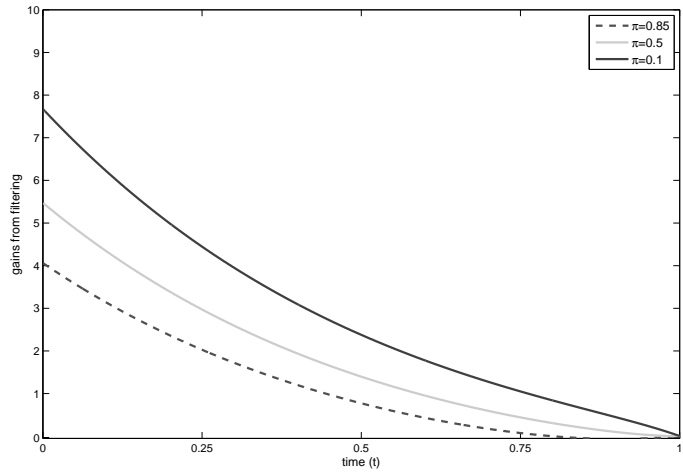}
\caption{Gains from filtering: $T=1$ year, $w=1$, $\rho=0$, $\vartheta=0.02$, $\lambda^+_1=10$, $\lambda^-_1=5$, $\lambda^+_2=5$, $\lambda^-_2=20$, $\theta=0.5$, $a_1=5$, $b_1=0.1$, $a_2=5$, $b_2=0.1$.}
\label{fig:gains}
\end{figure}

\section{Conclusion}
In this paper, we study the portfolio optimization problem for an investor who has an indirect effect on the risky asset prices. We represent the state of the market sentiment via a finite-state Markov chain, possibly not observable, whose generator depends on the portfolio choices of the investor. In this way we intend to model the influence of a large investor on the rest of the market, who tend to mimic her choices.

We solve the utility maximization problem from terminal wealth for an investor who is endowed with a logarithmic and a power utility under both full and partial information.
Under full information, from the mathematical point of view, we show that for logarithmic utility, pointwise maximization cannot be applied as, in this model, investors' decisions today may change the state of the market at a future date. Therefore we solve the problem by dynamic programming. In both logarithmic and power utility cases we show that the value function is the unique solution of the HJB equation, which reduces to a system of ODEs.

In the partial information case, we apply the reduction approach to transform the original optimization problem into an equivalent optimization problem where all state variables turn out to be observable with respect to the investor's filtration. In this setting, it is not possible to solve the problem directly. However, we can apply the theory of control for piecewise deterministic processes and show that the optimization problem has a solution and the optimal value function is the unique viscosity solution of the HJB equation.

We make a numerical study in a simpler example with a two-state Markov chain for a deeper understanding of the investor's optimal decisions. Interestingly we found that the behavior of the investor with an influence on the market is quite different if compared to the choices of an investor with no impact, and this has a return in gains from utility maximization, under both full and partial information. Moreover, our model allows for a scenario where the investor can make larger profits in bad market conditions than in good ones.  Finally, the use of filtering in the partial information, instead of average data, produces positive profits, that justify the additional complexity.

\appendix
\section{Technical proofs}\label{app:proofs}

\begin{proof}[Proof of Theorem \ref{thm:verification}]
The argument for the proof of this verification theorem is standard.
For part (i.), given an admissible control $h\in \mathcal H$, let $W^{(h)}$ be the solution to equation \eqref{wealthdyn} corresponding to the strategy $h$. let  $\{T_n\}_{n \ge 1}$ be the sequence of jump times of $Y$ and denote by $m$  the jump measure of $Y$,
\[
m([0,t]\times\{e_j\}):=\sum_{n \ge 1} \I_{\{Y_{T_n}=e_j\}}\I_{\{T_n\leq T\}}.
\]
Its compensator is then given by
\[
\phi([0,t]\times\{e_j\})=\int_0^t\sum_{i \neq j} q^{i,j}(h_s)\I_{\{Y_{s^-}=e_i\}}\ud s.
\]
Then the semimartingale decomposition of $Y^{(h)}$ is given by
\begin{align}
Y^{(h)}_t=&Y^{(h)}_0+\int_0^t\sum_{i, j=1}^K (e_j-e_i) q^{i,j}(h_s)\I{\{Y^{(h)}_{s^-}=e_i\}} \ud s \\
&+ \int_0^t \sum_{i,j=1}^K (e_j-e_i)\I_{Y^{(h)}_{s^-}=e_i}(m-\phi)(\ud s\times \{e_j\}), \quad t \in [0,T].\label{eq:semimgY}
\end{align}
Denoting the partial derivative of $\Upsilon$ with respect to time and wealth by $\Upsilon_t$ and $\Upsilon_w$, respectively and applying It\^{o}'s formula we get
\begin{align}
&\Upsilon(T, W_T^{(h)}, Y^{(h)}_T)=\Upsilon(t,w,e_i)+\int_t^T\mathcal L^h \Upsilon(s, W^{(h)}_{s^-},Y^{(h)}_{s^-})  \ud s\\
&+\int_t^T\sum_{i,j=1}^K\!\left(\Upsilon(s, W^{(h)}_s, e_j)-\Upsilon(s, W^{(h)}_s,  e_i)\I_{Y^{(h)}_{s^-}=e_i}\right) (m-\phi)(\ud s\times \{e_j\})\nonumber\\
&+ \int_\R \left(\Upsilon\left(s,W^{(h)}_{s^-}(1+h_{s} z), Y^{(h)}_{s}\right)-\Upsilon(s,W^{(h)}_{s^-},Y^{(h)}_{s})\right)\nu(\ud s, \ud z),\label{eq:G1}
\end{align}
where $\nu(\ud t, \ud z)$ is the compensated jump measure defined in \eqref{eq:nu}. Since $\Upsilon$ satisfies the HJB equation in \eqref{eq:HJB_1} we get
\begin{align*}
&\Upsilon(T, W_T^{(h)}, Y^{(h)}_T)\leq \Upsilon(t,w,e_i)\!\\
&+\!\!\int_t^T\!\!\sum_{i,j=1}^K\!\left(\Upsilon(s, W^{(h)}_{s^-},  e_j)\!-\!\Upsilon(s, W^{(h)}_{s^-}, e_i)\I_{Y^{(h)}_{s^-}=e_i}\right) (m-\phi)(\ud s\!\times\! \{e_j\})\\
&+\! \int_t^T\!\!\!\int_\R\! \left(\Upsilon\left(s,W^{(h)}_{s^-}(1+h_{s} z), Y^{(h)}_{s}\right)-\Upsilon(s,W^{(h)}_{s^-}, Y^{(h)}_{s})\right) \nu(\ud s, \ud z).
\end{align*}
By \eqref{eq:cond1} and \eqref{eq:cond2}, the stochastic integrals
\begin{gather*}
\int_0^t\!\!\sum_{i,j=1}^K\!\left(\Upsilon(s, W^{(h)}_{s^-},  e_j)\!-\!\Upsilon(s, W^{(h)}_{s^-}, e_i)\I_{Y^{(h)}_{s^-}=e_i}\right) (m-\phi)(\ud s\!\times\! \{e_j\}), \ t \in [0,T],\\
\int_0^t\!\!\!\int_\R\! \left(\Upsilon\left(s,W^{(h)}_{s^-}(1+h_{s} z), Y^{(h)}_{s}\right)-\Upsilon(s,W^{(h)}_{s^-}, Y^{(h)}_{s})\right) \nu(\ud s, \ud z), \ t \in [0,T],
\end{gather*}
are $(\bF, \P)$-true martingales (see, e.g. \cite[Theorem 26.12 part 2]{davis1993markov}). Hence, taking the expectation in \eqref{eq:G1} we obtain
\begin{align}
\Upsilon(t,w,e_i)\geq V(t,w,e_i).\label{eq:inequality}
\end{align}
For part $(ii.)$, if $h^*$ is a maximizer of equation \eqref{eq:HJB_1}, we get the equality in the expression \eqref{eq:inequality}.
\end{proof}

\medskip

\begin{proof}[Proof of Proposition \ref{prop:KS}]
Consider a function  $f:\mathcal E\to \R$. For every $\bF^S$-~predictable control $h$, using the semimartingale decomposition of $Y$ in equation \eqref{eq:semimgY} and applying the It\^{o}'s formula we get
\[
\ud f(Y_t)= Q^\top(h_t) f(Y_t) \ud t + \ud M^{(1)}_t,
\]
where $M^{(1)}$ is an $(\bF, \P)$-martingale and $Q^\top$ denotes the transpose of the generator matrix $Q$. Denote by $\widehat{H_t}$ the projection of $H_t$ over the $\sigma$-algebra $\F^S_t$, for some $\bF$-adapted process $H=\{H_t, \ t \in [0,T]\}$ i.e. $\mathbb E \left[H_t |\F^S_t\right]$, then we obtain
\[
\ud \widehat{f(Y_t)}= Q^\top(h_t) \widehat{f(Y_t)} \ud t + \ud M^{(2)}_t,
\]
where now $M^{(2)}$ is an $(\bF^S, \P)$-martingale. Using the martingale representation theorem for $(\bF^S, \P)$-martingale (see, for instance,  \cite[Theorem A5.5]{davis1993markov} ) we can write
\[
\widehat{f(Y_t)}-\widehat{f(Y_0)}-\int_0^t Q^\top(h_s) \widehat{f(Y_s)} \ud s=\int_0^t w^\pi(s,z) {\nu}^\pi(\ud s, \ud z),
\]
for some $\bF^S$-predictable process $w^\pi$ such that $\esp{\int_0^Tw^\pi(t,z) \eta^\P(t, e_i, \ud z)}<\infty$ for $i\in\{1, \dots, K\}$, where $\nu^\pi(\ud t, \ud z)$ is the $\bF^S$-compensated measure defined in \eqref{eq:compensator_partial}. Let $U_t=\int_0^t\int_{\R}C(s,z)\mu(\ud s, \ud z)$, for some $\bF^S$-predictable process $C$. Then we have
\[
\ud(U_t f(Y_t))= \left(U_t Q^\top(h_t) f(Y_t) +\int_\R f(Y_{t^-}) C(t,z) \eta^\P(t, Y_{t^-}, \ud z)\right)\ud t + \ud M^{(3)}_t,
\]
for some $(\bF, \P)$-martingale $M^{(3)}$. Projecting again over $\F^S_t$, and using the fact that $U$ is $\bF^S$-adapted, we get
\begin{align}\label{eq:projectio1}
\ud(\widehat{U_t f(Y_t)})= \left(U_t\widehat{Q^\top(h_t) f(Y_t)} +\int_\R \Gamma(t,z) \widehat{f(Y_{t^-})  \eta^\P(t, Y_{t^-}, \ud z)}\right)\ud t + \ud M^{(4)}_t,
\end{align}
where $M^{(4)}$ is an $(\bF^S, \P)$-martingale. Now we compute the product $U_t \widehat{f(Y_t)}$
\begin{align}
\ud(\widehat{U_t f(Y_t)})=& \left(U_t\widehat{Q^\top(h_t) f(Y_t)}+\int_\R \Gamma(t,z) w^\pi(t,z) \pi_{t^-}(\eta^\P(\ud z)) \right.\nonumber\\
&\left.+\int_\R C(t,z) \widehat{f(Y_{t^-})} \pi_{t^-}(\eta^\P(\ud z))\right)\ud t + \ud M^{(5)}_t. \label{eq:projectio2}
\end{align}
By the equality $U_t \widehat{f(Y_t)}=\widehat{U_tf(Y_t)}$, we get that the finite variation terms in equations \eqref{eq:projectio1} end \eqref{eq:projectio2} coincide, and this results to the expression for the process $w^\pi$
\[
w^\pi(t,z)=\frac{\ud \pi_{t^-}(f \eta^\P)}{\ud \pi_{t^-}(\eta^\P)}(z)-\pi_{t^-}(f), \quad (t,z)\in [0,T]\times \R,
\]
where $\frac{\ud \pi_{t^-}(f \eta^\P)}{\ud \pi_{t^-}(\eta^\P)}(z)$ is the Radon-Nikodym derivative of the measure $\pi_{t^-}(f \eta^\P(\ud z))$ with respect to $\pi_{t^-}(\eta^\P(\ud z))$. Finally choosing $f(Y_t)=\I_{\{Y_t=e_i\}}$ we get that
\[
w^\pi(t,z)=\pi_{t^-}^i \frac{1}{\sum_{j=1}^K \pi_t^j \frac{\ud \eta^{\P}(t,e_j,z)}{\ud \eta^{\P}(t,e_i,z)}}-1, \quad (t,z)\in [0,T]\times \R,
\]
where $\ds \frac{\ud \eta^{\P}(t,e_j,z)}{\ud \eta^{\P}(t,e_i,z)}$ is the Radon-Nikodym derivative of the measure $\eta^{\P}(t,e_j,\ud z)$ with respect to $\eta^{\P}(t,e_i,\ud z)$, which leads to Equation \eqref{eq:KS_I}.
\end{proof}

\medskip

\begin{proof}[Proof of Lemma \ref{lemma2.1}] The proof follows the same lines of \cite[Theorem 9.3.1]{bauerle2011markov}. Let $(T_n, Z_n)$ be the sequence of jump times and jump sizes of the PDMP. Then we have
\begin{align*}
V^{\{ h^n\}}&=\mathbb{E}^{\{ h^n\}}\left[U(W_T)\right]=\mathbb{E}^{\{ h^n\}}\left[\sum_{n=0}^\infty \I_{T_n<T<T_{n+1}}U(W_T)\right]\\
&=\sum_{n=0}^{\infty}\mathbb{E}^{\{ h^n\}}\left[\mathbb{E}^{\{h^n\}}\left[ \I_{T_n<T<T_{n+1}}U(W_T)|T_n<T, X_{T_n\wedge T}\right]\right]\\
&=\sum_{n=0}^{\infty}\mathbb{E}^{\{ h^n\}}\left[\mathbb{E}^{\{h^n\}}\left[ e^{-\Lambda^{h^n}_{T-T_n}(L_n)}U(w_{T-T_n})\right]\I_{T_n<T}\right]\\
&=\mathbb{E}^{\{h^n\}}\left[\sum_{n=0}^\infty \I_{T_n<T}r(L_n, h^n)\right]=J^{\{h^n\}}_\infty.
\end{align*}
\end{proof}

\medskip

\begin{proof}[Proof of Lemma \ref{lemma3.3}] Since $e^{-\Lambda^{\alpha}_u(\widetilde{x})}<1$ we get $r(\widetilde{x},\alpha)\le w$. Next we turn to estimating $Q_L b(\widetilde x, \alpha)=\int_{\widetilde {\mathcal{X}}}b(x')Q_L(\ud x'|\widetilde{x},\alpha)$. It holds that
\begin{align*}
&\int_{\widetilde{\mathcal{X}}}b(x')Q_L(\ud x'|\widetilde{x},\alpha) \\
&\qquad = \int_0^{T-t}e^{c(T-s-t)}e^{-\Lambda^{\alpha}_s(\widetilde x)}\int_{-L}^L\int_{\mathbb{R}}w(1+hz)\sum_{j=1}^K\pi_j\eta^j(t+s,\ud z)\alpha_s(\ud h) \ud s\\
&\qquad\le b(\widetilde{x})c_{\eta}\int_0^Te^{-c r}\ud r=b(\widetilde{x})c_{\eta}\frac{1}{\gamma}(1-e^{-c T}) \le \frac{c_{\eta}}{\gamma}b(\widetilde{x}),
\end{align*}
where we define $$c_{\eta}=\underset{t\in[0,T]}{\underset{j\in\{1,\dots, K\}}{\underset{h \in [-L,L]}{\sup}}}\left\{ \int_{\mathbb{R}}(1+hz)\eta^\P(t,e_j,\ud z) \right\}<\infty.$$ Clearly $\frac{c_{\eta}}{c}<1$ for sufficiently large $c$, so that the MDM is contracting.
\end{proof}

\medskip



\begin{proof}[Proof of Proposition \ref{prop:continuity}]
Let $(\widetilde x_n, \alpha_n)$ be a sequence converging to $(\widetilde x, \alpha)$ as $n \to \infty$. Then by \cite[Theorem 43.5]{davis1984piecewise} we have that
\[
\lim_{n \to \infty}\sup_{u \in [0,T]} |\widetilde \varphi^{\alpha_n}_u(\widetilde x_n)-\widetilde \varphi^\alpha_u(\widetilde x)|=0.
\]
This implies the continuity of the reward function $r$. Moreover the continuity of the mapping $(\widetilde{x},\alpha) \mapsto Q_L v (\widetilde{x},\alpha)$ follows from the fact that, for every function $v \in \mathcal B_b$, by Assumption \ref{ass3} the mapping
\[
(\widetilde x, \alpha)\mapsto\int_{\R}v(t, w(1+hz), \pi^1(1+u^1(t,\bpi, z)), \dots, \pi^1(1+u^1(t,\bpi, z)) )\eta^\P(t, e_i, \ud z)
\]
is continuous. To prove this we can apply for instance \cite[Lemma A5]{colaneri2016shall} since, in our setting,  $t \mapsto \eta^\P(t, e_i, z)$ is continuous and $\lambda^{max}:=\underset{t\in[0,T]}{{\underset{i \in \{1, \dots, K\}}{\sup}}}\eta^\P(t, e_i, \ud z)<\infty$.
\end{proof}

\section*{Acknowledgments}
The authors thank the participants of Vienna Seminar in Mathematical Finance and Probability and Brown Bag Seminar of WU Vienna University of Business and Economics, especially to R\"{u}diger Frey and Uwe Schmock. The authors thank Claudia Ceci for carefully reading the article and making many valuable suggestions for improvement.
 S\"{u}han Altay gratefully acknowledges partial financial support from the Austrian Science Fund (FWF) under grant P25216.
The work on this paper was completed while Zehra Eksi was visiting the Department of Economics, University of Perugia as a part of the ACRI Young Investigator Training Program (YITP). The support of the Association of Italian Banking Foundations and Savings Banks (ACRI) is greatly acknowledged.

\bibliographystyle{plainnat}
\bibliography{biblio}

\end{document}